\documentclass[journal,twoside]{IEEEtran}
\usepackage{subfigure}
\usepackage{graphicx}
\DeclareGraphicsExtensions{.eps,.pdf}  
\usepackage{epstopdf}
\graphicspath{ {figures/} }
\usepackage{cite}
\usepackage{amsmath,amssymb,amsfonts}
\usepackage{algorithmic}
\usepackage{textcomp}
\usepackage{xcolor}
\usepackage{algorithm}
\usepackage{algorithmic}
\usepackage{float}
\usepackage{amsthm}
\usepackage{array}

\usepackage{cases}
\usepackage{etoolbox}
\makeatletter
\newbox\boxzero
\patchcmd{\numcases}{\setbox\z@}{\setbox\boxzero}{}{}
\patchcmd{\numcases}{\kern\wd\z@}{\kern\wd\boxzero}{}{}
\patchcmd{\endnumcases}{\displaystyle \box\z@}{\displaystyle\box\boxzero}{}{}

\makeatother

\newtheorem{remark}{\bf Remark}

\newtheorem{lemma}{\bf \underline{Lemma}}
\newcommand{\maxtt}{\mathtt{max}}
\newcommand{\hSP}{h_{\mathtt{SP}}}
\newcommand{\hSS}{h_{\mathtt{SS}}}
\newcommand{\hSE}{h_{\mathtt{SE}}}
\newcommand{\gUP}{g_{\mathtt{UP}}}
\newcommand{\gUE}{g_{\mathtt{UE}}}
\newcommand{\gUS}{g_{\mathtt{US}}}
\newcommand{\SP}{\mathtt{SP}}
\newcommand{\ttSS}{\mathtt{SS}}
\newcommand{\SE}{\mathtt{SE}}
\newcommand{\UP}{\mathtt{UP}}

\newcommand{\E}{\mathtt{E}}
\newcommand{\UAV}{\mathtt{U}}
\newcommand{\Ss}{\mathtt{S}}
\newcommand{\Pp}{\mathtt{P}}
\newcommand{\SEC}{\mathtt{sec}}
\newcommand{\st}{\mathtt{s. t.}}

\begin{document}

\title{\LARGE UAV-Assisted Secure Communications in Terrestrial Cognitive Radio Networks: Joint Power Control and 3D Trajectory Optimization\\
}
\author{
\IEEEauthorblockN{Phu X. Nguyen, Van-Dinh Nguyen, Hieu V. Nguyen,  and Oh-Soon Shin}\\

\thanks{Part of this paper was presented at the IEEE Consumer Communications \& Networking Conference (CCNC), Las Vegas, USA, in Jan. 2019 \cite{CCNC}.}
\thanks{P. X. Nguyen, H. V. Nguyen, and O.-S. Shin are with  Soongsil University, Seoul 06978, South Korea (e-mail: nxphu.1994@gmail.com; hieuvnguyen@ssu.ac.kr; osshin@ssu.ac.kr).}
\thanks{V.-D. Nguyen is with SnT – University of Luxembourg, L-1855 Luxembourg. He was also with  Soongsil University, Seoul 06978, South Korea (email: dinh.nguyen@uni.lu). }
}

\maketitle

\begin{abstract}
This paper considers secure communications for an underlay cognitive radio network (CRN) in the presence of an external eavesdropper (Eve). The secrecy performance of CRNs is usually limited by the primary receiver's interference power constraint. To overcome this issue, we propose to use an unmanned aerial vehicle (UAV) as  a friendly jammer to interfere Eve in decoding the confidential message from the secondary transmitter (ST). Our goal is to jointly optimize the transmit power and UAV's trajectory in the three-dimensional (3D) space to maximize the average achievable secrecy rate of the secondary system. The formulated optimization problem is nonconvex due to the nonconvexity of the objective and non-convexity of constraints, which is very challenging to solve. To obtain a suboptimal but efficient solution to the problem, we first transform the original problem into a more tractable form and develop an iterative algorithm for its solution by leveraging the inner approximation framework. We further extend the proposed algorithm to the case of imperfect location information of Eve, where the average worst-case secrecy rate is considered as the objective function. Extensive numerical results are provided to demonstrate the merits of the proposed algorithms over existing approaches.																																																																																																											
\end{abstract}


\section{Introduction}
Recently, the rapidly increasing number of mobile devices and multimedia services have made radio spectrum scarce and expensive resource \cite{FCC, Tehrani, Song_Cog}. To exploit spectrum more efficiently, cognitive radio has been widely considered as a promising solution \cite{TragosCST13}, which enables to learn the surrounding context and to adjust the operating parameters, thereby adapting to  changes of radio frequency environment. Accordingly, secondary devices are allowed to use the licensed bands simultaneously, making cognitive radio a potential approach for future wireless networks. However, various malicious wireless devices can also opportunistically access the licensed spectrum, which might make cognitive radio networks (CRNs) vulnerable \cite{Z_Shu, Sharma, F_Zhu, Nguyen:TIFS:16,NguyenTCCN17}. For instance, when a secondary transmitter (ST) transmits confidential messages to a secondary receiver (SR), an external eavesdropper (Eve, also known as a passive attacker) probably overhears and intercepts the legitimate transmissions. 

Traditionally, the complexity-based cryptography can be effective when the computational ability of Eves is too restricted to decipher  secret key. Nevertheless,  Eve's computing power is evolving consistently, while a trust infrastructure for guaranteeing confidential communications is expensive to deploy. To overcome such challenges, physical-layer security (PLS)  has been introduced as a potential technique to prevent eavesdropping without a secure cryptographic protocol \cite{2, Wyner}. The key idea of PLS is to exploit random characteristics of the wireless channel to degrade the Eve's decoding capability. To make the PLS viable, jamming noise (JN) can be embedded at the transmitter and transmitted along with the information signals  to degrade the channel quality of  Eve \cite{NguyenTCCN17, Bouabdellah, NLi_XTao}. A large effort has been made  to bring the PLS a step closer to practice \cite{WuIT16,NguyenJSAC18, Bassily, Q_Li}. However, most of the conventional JN-based schemes are based on the ground jammers at the fixed locations, leading to several major challenges. First, when jammers are set far away from Eves, the effect of JN  is significantly reduced, and thus the secrecy rate deteriorates. Second, for JN to be effective, the legitimate transmitter needs to be aware of the channel state information (CSI) between itself and Eve. Since Eves are usually passive, it may not be possible to obtain their instantaneous CSI. Finally, in CRNs, the secrecy performance improvement of the secondary system using JN may also affect the primary system; the interference power to the primary receiver (PR) may exceed the predefined threshold.

In recent years, unmanned aerial vehicle (UAV) has attracted significant interest in many applications, such as agriculture, traffic control, military, photography, and package delivery \cite{YengComM17, M_Mozaffari, Gupta, Relaying_Systems, Sotheara}. PLS can benefit from the application of UAV as well, by making UAV send a JN to Eves. Compared with the on-ground jammer, there are two obvious advantages of UAV-aided JN: $i$) Eve will  undergo strong interference due to the line-of-sight (LoS) dominated UAV-Eve channel; $ii$) A UAV operating in the three-dimensional (3D) space at the altitude of a few hundred meters is able to fly to an optimal location to cause interference to the channel between ST and Eve by emitting a friendly JN. Thus,  it is expected that UAV-aided JN can provide better secrecy performance as compared to the
conventional on-ground jamming.

\subsection{Related Works}
PLS of CRNs has been well studied recently, which dealt with specific security risks due to the broadcasting nature of the wireless transmission media \cite{Z_Shu, Sharma, F_Zhu, Nguyen:TIFS:16,NguyenTCCN17,11,NguyenTVTPLS}. In general, these works mainly focused on secure communications for the secondary system \cite{Z_Shu, Sharma, F_Zhu, Nguyen:TIFS:16} and the primary system \cite{11,NguyenTVTPLS}, where power control is an effective way to control the interference, assuming that the CSI of the ST-PR links is already known. In \cite{NguyenTCCN17}, a cooperative transmission strategy was proposed to maximize the
minimum secrecy rate of the secondary system while satisfying the minimum secrecy rate achievable for the primary system. The common technique used in the above works is to make JN and the desired signal concurrently transmitted at the same transmitter (ST or primary transmitter), which limits the effectiveness of JN. The transmitter needs to be equipped with multiple antennas to perform beamforming; otherwise the  legitimate user must have better channel condition than Eve, which is too optimistic in practice.

%

The security performance of ground users in the presence of a ground Eve is improved by using UAV as a mobile relaying \cite{12}. In this work,  UAV is assumed to fly with a fixed trajectory, leading to a suboptimal solution. In \cite{6}, a UAV is used to transmit a friendly JN with the aim of interfering the channel between ST and Eve, where the security performance of UAV-to-ground communication is maximized by jointly optimizing  the UAV trajectory and the transmission power. The authors in \cite{14} proposed  a cooperative jamming UAV to enable confidential air-to-ground communications between a mobile UAV and ground nodes, where the user scheduling, UAV's trajectory and the transmit power are jointly optimized to maximize the minimum secrecy rate among ground nodes. In general, the location information of Eves is assumed to be perfectly known \cite{12, 6, 13, 14}. A practical scenario was considered in \cite{ImperfectLocation} in which the location information of Eves is unknown. Notably, the trajectories of UAV in the 3D space were not considered in \cite{ImperfectLocation}, presumably due to the nonconvexity and complexity of the constraints related to UAV mobility.
\subsection{Main Contribution}
In this paper, we study the PLS for CRNs, in which the secure communication of  secondary system is guaranteed by using a UAV as a friendly jammer.  UAV is controlled to move in a period of time that consists of many intervals, called time slots. Such a time slot is designed to be suitable with the motion characteristics of UAV in the 3D space. We first formulate the average achievable secrecy rate maximization problem over all time slots, where UAV's trajectory and power allocation are jointly optimized under the transmit power constraints,  interference power at the PR caused by both UAV and ST, and mobility capability of UAV. The formulated problem is highly nonconvex due to strong coupling between optimization variables, which makes hard find the globally optimal solution. The methods used in \cite{6,13,14, ImperfectLocation} mainly utilize the inner convex approximation to tackle subproblems. Herein, each subproblem is a single variable optimization problem, which is divided from the original optimization problem. Such an approach often results in a slow convergence rate, and yet, the convergence of these proposed heuristic method is not theoretically guaranteed.

To the best of our knowledge, our earlier work in \cite{CCNC} is the first work that aims at improving the secrecy rate of the on-ground secondary system by using UAV-enabled cooperative JN. Differently from \cite{CCNC}, this paper considers the following completely new issues: $i)$ We aim at finding the optimal trajectory of UAV in the 3D space instead of the two-dimensional (2D) space, by jointly optimizing its altitude as well as horizontal location; $ii)$ Towards a realistic scenario, the imperfect location information of Eve  is also considered, making the problem even more challenging to solve. As a result, the main
contributions of the paper are summarized as follows.
\begin{itemize}
\item We propose a new model for PLS in CRNs to maximize the average achievable secrecy rate of the
secondary system  by exploiting UAV-enabled JN.
\item We formulate a new optimization problem that jointly optimizes the transmit power and UAV's trajectory subject to the PR's interference power constraint.  We first consider the perfect CSI, including Eve, to investigate benefits of our new model, for which an efficient and low-complexity algorithm is proposed. The key idea of our approach is to transform the original nonconvex problem into a more tractable form and then develop  new inner approximations (IAs) of
 nonconvex parts \cite{Marks:78,Beck:JGO:10}, which guarantees convergence
at least to a locally optimal solution. 
\item When the location information of Eve is imperfect and Eve is assumed to be distributed in a circular region with a given radius, we reformulate the optimization problem by considering  the worst-case secrecy rate. The main difficulty of this problem comes from the rate function of Eve, which is further shaped to have a set of convex constraints by combining tools from IA framework
and $S$-procedure.
\item Extensive numerical results are provided to demonstrate that the proposed
algorithms have low complexities, i.e., in terms of per-iteration computation and the number of iterations, and to show great performance improvement over existing
schemes. Numerical results also confirm the effectiveness of the proposed approach
that optimizes the altitude of Eve as well as the horizontal location.
\end{itemize}


\subsection{Paper Organization and Notation}
The remainder of this paper is organized as follows. The system model is introduced in Section II. The optimization problems and the proposed algorithms under perfect and imperfect location information of Eve are provided in Section III and Section IV, respectively. Numerical results are given in Section V. Finally, Section VI concludes the paper.

\textit{Notation}: Bold lower and upper case letters denote vectors
and matrices, respectively. $\mathbb{E}\{\cdot\}$ represents the expectation of random variables. $\nabla$ denotes the gradient of a function. The superscript $(\cdot)^{T}$ denotes the transpose of a matrix. $\mathbf{A}  \succeq 0$ indicates that $\mathbf{A}$ is a positive semidefinite matrix. $\left<\mathbf{a},\mathbf{b}\right>$ is the inner product of two vectors $\mathbf{a}$ and $\mathbf{b}$. $\ln(X)$ denotes the natural logarithm of $X$.

\section{System Model}

\begin{figure}[!ht]
	\begin{center}
		\includegraphics[width=0.8\columnwidth] {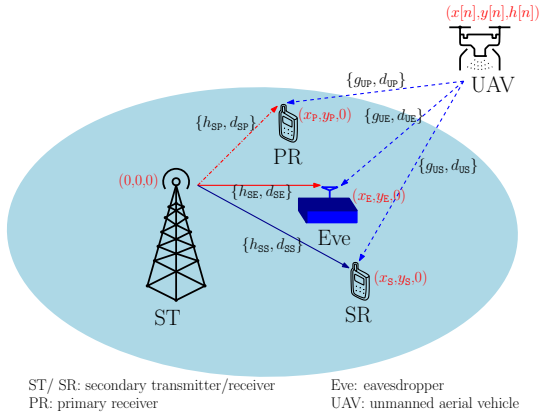}
	\end{center}
	\caption{Illustration of a CRN with a UAV-aided JN and an external Eve.\label{fig:systemmodel}}
\end{figure}

\subsection{Transmission Model}
 We consider an on-ground CRN consisting of an ST and an SR in the presence of a PR and an Eve, as illustrated in Fig. \ref{fig:systemmodel}. Herein, Eve endeavors to intercept and overhear the legitimate transmission between  ST and  SR in the secondary network. In order to further enhance the PLS of CRN, we propose to use  UAV as a friendly jammer   to degrade the eavesdropping channel. Let us define the 3D space $\mathcal{T}\triangleq\{(x,y,z)|x,y,z\in\mathbb{R}\}$. The positions of ground nodes (ST, SR, PR and Eve) in the 3D-space model are expressed as $\mathbf{c}_{\mathtt{ST}}\triangleq\left(0,0,0\right)$, $\mathbf{c}_{\mathtt{S}}\triangleq\left(x_\Ss,y_\Ss,z_\Ss\right)$, $\mathbf{c}_{\mathtt{P}}\triangleq\left(x_\Pp,y_\Pp,z_\Pp\right)$, and $\mathbf{c}_{\mathtt{E}}\triangleq\left(x_\E,y_\E,z_\E\right)$, respectively. Herein, the SR, PR and Eve are located on the ground, i.e., $z_\Ss = z_\Pp = z_\E = 0$. 

The predefined time interval $ T $ of UAV  is split into $N$ time slots of equal length, i.e.,  the duration of each time slot is given as $\delta_t=T/N$. Note that $ N $ must be large enough to guarantee a  small interval per time slot, such that in each time slot the UAV's location is almost unchanged. Thus, we define the time-varying horizontal coordinate of UAV as $\mathbf{c}_{\mathtt{U}}[n]\triangleq\left(x_{\mathtt{U}}\left[n\right], y_{\mathtt{U}}\left[n\right], z_{\mathtt{U}}\left[n\right]\right)$, $\forall n\in\mathcal{N}\triangleq\{1,2,\cdots,N\}$, where the altitude of UAV is limited in the range $h^{\min} \leq z_{\mathtt{U}}\left[n\right] \leq h^{\max}$. The UAV is assumed to move from the initial position $\mathbf{c}_{\mathtt{U}}[0]\triangleq\left(x_{0},y_{0},h_{0}\right)$ to the final predefined position $\mathbf{c}_{\mathtt{U}}[N+1]\triangleq\left(x_{f},y_{f},h_{f}\right)$.  Furthermore, the maximum velocity constraint can be formulated as
$\Vert \bar{\mathbf{q}}_{\mathtt{v}}'(t) \Vert \leq V_{\max},\enspace 0 \leq t \leq T$, where $\bar{\mathbf{q}}_{\mathtt{v}}'(t)$  and $V_{\max}$ are the derivative of the UAV's position with respect to $t$ and the maximum speed of UAV, respectively. Accordingly, for a small interval of time slot $\delta_t$, the mobility constraints of UAV can be expressed as
\begin{IEEEeqnarray}{lll}\label{mobility}
	&h^{\min} \leq z_{\mathtt{U}}\left[n\right] \leq h^{\max},\ \forall n\in\mathcal{N}, \IEEEyessubnumber\label{1}\\
	&f_{\mathtt{d}}(\mathbf{c}_{\mathtt{U}}[n],\mathbf{c}_{\mathtt{U}}[n-1]) \leq L^{2}_{\max},\ \forall n\in\mathcal{N},\qquad \IEEEyessubnumber\label{2}\\
	&f_{\mathtt{d}}(\mathbf{c}_{\mathtt{U}}[N+1],\mathbf{c}_{\mathtt{U}}[N])=0,\IEEEyessubnumber\label{4}\qquad
\end{IEEEeqnarray}
where $ L_{\max}\triangleq V_{\max}\delta_t $ and $ f_{\mathtt{d}}(\mathbf{a},\mathbf{b}) \triangleq (x_{a}-x_{b})^{2}+(y_{a}-y_{b})^{2}+(z_{a}-z_{b})^{2} $, with $ \mathbf{a}\triangleq(x_{a},y_{a},z_{a})$ and $ \mathbf{b}\triangleq(x_{b},y_{b},z_{b})\in\mathcal{T} $.

\subsection{Achievable Secrecy Rate}
We assume that the air-to-ground channels are modeled as LoS channels. The  distances between  UAV and ground nodes are calculated as $ d_{\mathtt{Ux}}[n] \triangleq f_{\mathtt{d}}(\mathbf{c}_{\mathtt{x}},\mathbf{c}_{\mathtt{U}}[n])  $, for $\mathtt{x}\in\{\mathtt{S},\mathtt{P},\mathtt{E}\}$. At the time slot $n$, the channel gains from the UAV to SR, PR and Eve,  denoted by $\gUS$, $\gUP$ and $\gUE$, respectively, can be modeled according to the free-space path loss \cite{Relaying_Systems, 12, 6,13,14, ImperfectLocation}, i.e., $g_{\mathtt{Ux}}[n]=\rho_{0}(d_{\mathtt{Ux}}[n])^{-2}$, where $\rho_{0}$ is  the channel gain at the reference distance $d_{0}=1$ m. The terrestrial channels experience quasi-static independent 
 Rayleigh fading. Therefore, the channel gains of the links from the ST to SR, PR and Eve, denoted by $\hSP, \hSS$ and $\hSE$, respectively, can be expressed as $h_{\mathtt{Sx}}=\rho_{0}(d_{\mathtt{Sx}})^{-\varphi}\psi_{\mathtt{Sx}}$, where $ d_{\mathtt{Sx}}\triangleq f_{\mathtt{d}}(\mathbf{c}_{\mathtt{x}},\mathbf{c}_{\mathtt{ST}}) $;  $\varphi$ and $\psi_{\mathtt{Sx}}$ are the path loss exponent and an exponential random variable with unit mean, respectively.

 The achievable rates at SR and Eve for decoding the messages from ST at the time slot $n$ can be expressed as \cite{CCNC, 6}
\begin{IEEEeqnarray}{rCl}
R_{\Ss}[n]&=&\mathbb{E}_{\hSS}\Bigl\{\log_{2}\Bigl(1+\dfrac{p_{\Ss}[n]\hSS}{p_{\UAV}[n]\gUS[n]+\sigma^{2}}\Bigr)\Bigr\},\IEEEyessubnumber\label{15}\\
R_{\E}\left[n\right]&=&\mathbb{E}_{\hSE}\Bigl\{\log_{2}\Bigl(1+\dfrac{p_{\Ss}[n]\hSE}{p_{\UAV}[n]\gUE[n]+\sigma^{2}}\Bigr)\Bigr\},\IEEEyessubnumber\label{16}
\end{IEEEeqnarray}
where $p_{\Ss}[n]$ and $p_{\UAV}[n]$ are the transmit powers at the ST and UAV, respectively, and $\sigma^{2}$ is the power of additive white Gaussian noise (AWGN). For total $N$ time slots,  the average achievable secrecy rate for the secondary system can be expressed as \cite{15}
\begin{align}
	R_{\SEC}\triangleq\dfrac{1}{N}\sum_{n\in\mathcal{N}}\bigl[R_{\Ss}[n]-R_{\E}[n]\bigl]^{+},
\end{align}
where $[x]^+\triangleq\max\{0,x\}$.

\section{Proposed Algorithm with Perfect Location Information of Eavesdropper} \label{Alg Under Perfect}
In this section, the Eve's location information is
assumed to be perfectly known at the transmitters (ST and UAV). This assumption is of interest in some scenarios. For instance, at the beginning of the time interval, both SR and Eve perform handshaking with ST by sending pilot signals. However, only SR is scheduled to be served, while Eve is treated as an untrusted user. In addition, the system performance under the assumption of  perfect location information of Eve will act as an upper bound for the practical system, providing a reference of the potential benefit of using UAV-aided JN.

\subsection{Optimization Problem Formulation}
In this paper, the key idea is to exploit the advantage of UAV's mobility in combination with developing an effective power control scheme to enhance the security performance of the secondary system while satisfying the transmit power constraints and the PR's interference power constraint. By defining $\mathbf{c}\triangleq\{\mathbf{c}_{\mathtt{U}}[n]\}_{n\in\mathcal{N}}$ and  $\mathbf{p}\triangleq \{p_{\Ss}[n],p_{\UAV}[n]\}_{n\in\mathcal{N}}$, the secrecy rate maximization (SRM) problem for the secondary system is formulated as follows:
\begin{IEEEeqnarray}{lCl}\label{originalproblem}
\underline{\mathtt{P}:}\hspace{1em}\underset{\mathbf{c},\mathbf{p}}{\maxtt}&& \enspace R_{\SEC}\IEEEyessubnumber\label{originalproblem:a}\\
\st\quad&& \eqref{mobility} \IEEEyessubnumber\label{originalproblem:b},\\
&&\dfrac{1}{N}\sum_{n\in\mathcal{N}}p_{\UAV}\left[n\right] \leq \bar{P_{\UAV}},\IEEEyessubnumber\label{originalproblem:c}\\
&&0 \leq p_{\UAV}\left[n\right] \leq P_{\UAV}^{\max},\ \forall n\in\mathcal{N},\IEEEyessubnumber\label{originalproblem:d}\\
&&\dfrac{1}{N}\sum_{n\in\mathcal{N}}p_{\Ss}\left[n\right] \leq \bar{P_{\Ss}},\IEEEyessubnumber\label{originalproblem:e}\\
&&0 \leq p_{\Ss}\left[n\right] \leq P_{\Ss}^{\max},\ \forall n\in\mathcal{N},\IEEEyessubnumber\label{originalproblem:f}\\
&&\dfrac{1}{N}\sum_{n\in\mathcal{N}}\left(\mathbb{E}_{\hSP}\left\{p_{\Ss}[n]\hSP\right\}+p_{\UAV}[n]\gUP[n]\right)\leq \varepsilon.\qquad\IEEEyessubnumber\label{originalproblem:g}
\end{IEEEeqnarray}
Constraints \eqref{originalproblem:c} and \eqref{originalproblem:d} are the average power and the peak power constraints at  UAV, respectively. The average power and the peak power constraints at ST are stated  by \eqref{originalproblem:e} and \eqref{originalproblem:f}, respectively. Herein, we assume that $\bar{P_{\Ss}}\leq P_{\Ss}^{\max}$ and $\bar{P_{\UAV}}\leq P_{\UAV}^{\max}$. To guarantee the quality of service (QoS) of the primary system, the average power of aggregated interference at PR is limited by a predefined threshold $\varepsilon$ as in \eqref{originalproblem:g}.
It is not difficult to see that the objective function \eqref{originalproblem:a} is nonconcave and constraint \eqref{originalproblem:g} is nonconvex. Strong coupling between the optimization variables makes the problem even more challenging to be tackled. Moreover, the objective function may not be addressed directly due to the expectation of  the average achievable secrecy rate. In what follows, we first transform problem \eqref{originalproblem} into a more tractable form by bypassing the expectation functions with respect to the ground channels. Then, a low-complexity iterative algorithm based on IA framework is developed to solve the problem, which yields at least  a locally optimal solution.

\subsection{Tractable Formulation for \eqref{originalproblem}}
In the PLS, it is important to consider a safe design, taking into account the effects of wireless channels.  To do so, we derive a lower bound of $R_{\Ss}[n]$ and an upper bound of $R_{\E}[n]$ following the similar developments in \cite{6}.

\textit{Lower bound of} $R_{\Ss}[n]$:
Let $X[n]\triangleq\dfrac{p_{\Ss}[n]\hSS}{p_{\UAV}[n]\gUS[n]+\sigma^{2}}.
$
Since $h_{\mathtt{SS}}=\rho_{0}(d_{\mathtt{SS}})^{-\varphi}\psi_{\mathtt{SS}}$, we have
$
	X[n]=\dfrac{p_{\Ss}[n]\rho_{0}(d_{\mathtt{SS}})^{-\varphi}\psi_{\mathtt{SS}}}{p_{\UAV}[n]\gUS[n]+\sigma^{2}}.
$
It is true that $X[n]$ is an exponentially distributed random variable with parameter
$
	\lambda_\Ss[n]=d_{\mathtt{SS}}^{\varphi}/\Bigl(\dfrac{p_{\Ss}[n]\rho_{0}}{p_{\UAV}[n]\gUS[n]+\sigma^{2}}\Bigr).
$
$R_{\Ss}[n]$ in \eqref{15} can be rewritten as
\begin{IEEEeqnarray}{rCl} \label{19_app}
	R_{\Ss}[n]&=&\mathbb{E}_{\hSS}\bigl\{\log_{2}\bigr(1+X[n]\bigl)\bigr\}\nonumber\\
	          &=&\mathbb{E}_{\hSS}\Bigl\{\log_{2}\bigr(1+e^{\ln(X[n])}\bigl)\Bigr\}.
\end{IEEEeqnarray}
Since $\log_{2}\bigr(1+e^{x}\bigl)$ is a convex function \cite{Stephen} and by Jensen's inequality, it follows that
\begin{IEEEeqnarray}{rCl} \label{20_app}
	R_{\Ss}[n]&=&\mathbb{E}_{\hSS}\Bigl\{\log_{2}\bigr(1+e^{\ln(X[n])}\bigl)\Bigr\} \nonumber\\
	          &\geq&\log_{2}\bigr(1+e^{\mathbb{E}_{\hSS}\{\ln(X[n])\}}\bigl),
\end{IEEEeqnarray}
where $\mathbb{E}_{\hSS}\bigl\{\ln\bigr(X[n]\bigl)\bigr\}$ is computed as
\begin{IEEEeqnarray}{rCl} \label{21_app}
	\mathbb{E}_{\hSS}\Bigl\{\ln\bigr(X[n]\bigl)\Bigr\} &=& \int_0^\infty \ln\bigr(X[n]\bigl)\lambda_\Ss[n]e^{-\lambda_\Ss[n]x}dx \nonumber\\
	&=& -\ln\bigr(\lambda_\Ss[n]\bigl) - k,
\end{IEEEeqnarray}
with $k$ being the Euler constant.
Substituting \eqref{21_app} into \eqref{20_app}, we get
\begin{IEEEeqnarray}{rCl} \label{22_app}
R_{\Ss}\left[n\right]\geq R_{\Ss}^{\mathtt{LB}}\left[n\right]\triangleq\log_{2}\bigl(1+\dfrac{e^{-k}\gamma_{0}d_{\ttSS}^{-\varphi}p_{\Ss}\left[n\right]}{\gamma_{0}d^{-2}_{\mathtt{US}}[n]p_{\UAV}[n]+1}\bigr),
\end{IEEEeqnarray}
where $\gamma_{0}\triangleq\rho_0/\sigma^{2}$.

\textit{Upper bound of} $R_{\E}[n]$:
Since $h_{\mathtt{SE}}=\rho_{0}(d_{\mathtt{SE}})^{-\varphi}\psi_{\mathtt{SE}}$, we have
\begin{IEEEeqnarray}{rCl} \label{24_app}
	Y[n]\triangleq\dfrac{p_{\Ss}[n]\hSE}{p_{\UAV}[n]\gUE[n]+\sigma^{2}} = \dfrac{p_{\Ss}[n]\rho_{0}(d_{\mathtt{SE}})^{-\varphi}\psi_{\mathtt{SE}}}{p_{\UAV}[n]\gUE[n]+\sigma^{2}}.
\end{IEEEeqnarray}
Similarly to $X[n]$, $Y[n]$ is also an exponentially distributed random variable with parameter
$
	\lambda_\E[n]=d_{\mathtt{SE}}^{\varphi}/\Bigl(\dfrac{p_{\Ss}[n]\rho_{0}}{p_{\UAV}[n]\gUE[n]+\sigma^{2}}\Bigr).
$
Given that $\log_{2}\bigr(1+y\bigl)$ is a concave function in $y$ \cite{Stephen}, we obtain the following inequality using Jensen's inequality:
\begin{IEEEeqnarray}{rCl} \label{27_app}
	R_{\E}[n]&=&\mathbb{E}_{\hSE}\bigl\{\log_{2}\bigr(1+Y[n]\bigl)\bigr\} \nonumber \\
	&\leq&\log_{2}\bigr(1+\mathbb{E}_{\hSE}\{Y[n]\}\bigl)\nonumber\\
	&=&R_{\E}^{\mathtt{UB}}[n]\triangleq\log_{2}\Bigl(1+\dfrac{\gamma_{0}d_{\SE}^{-\varphi}p_{\Ss}\left[n\right]}{\gamma_{0}d^{-2}_{\mathtt{UE}}[n]p_{\UAV}[n]+1}\Bigr),
\end{IEEEeqnarray}
where $	\mathbb{E}_{\hSE}\bigl\{\ln(Y[n])\bigr\} = 1/\lambda_\E[n].$

In addition, from the fact that $\mathbb{E}_{\hSP}\left\{p_{\Ss}[n]\hSP\right\}=\rho_0 d_{\SP}^{-\varphi}p_{\Ss}[n]$, constraint \eqref{originalproblem:g} can be further simplified as
	\begin{equation}\label{interpowerPR}
	\dfrac{1}{N}\sum_{n\in\mathcal{N}}\left(\rho_0 d_{\SP}^{-\varphi}p_{\Ss}[n]+ \rho_0d^{-2}_{\UP}[n]p_{\UAV}[n]\right)\leq \varepsilon.
	\end{equation}

Simply put, we consider the following safe optimization
problem:
\begin{IEEEeqnarray}{rcll}\label{problemLB}
	\underline{\mathtt{P}^{\mathtt{Safe}}:} &\quad \underset{\mathbf{c},\mathbf{p}}{\maxtt}&& \quad R_{\SEC}^{\mathtt{LB}}\triangleq \dfrac{1}{N}\sum_{n\in\mathcal{N}}\Bigl(R_{\Ss}^{\mathtt{LB}}[n]- R_{\E}^{\mathtt{UB}}[n]\Bigr)\IEEEyessubnumber\label{problemLB:a}\quad\\
	&\quad\st&& \quad \eqref{mobility}, \eqref{originalproblem:c}-\eqref{originalproblem:f}, \eqref{interpowerPR}, \IEEEyessubnumber\label{problemLB:b}
\end{IEEEeqnarray}
where the operation $[x]^+$ is ignored since it does not affect the optimal solution. If the objective function is less than zero for any time slot, ST can reduce its transmit power of ST to zero while satisfying constraint \eqref{interpowerPR}.
\begin{remark}
Note that problem \eqref{problemLB} is considered as a safe design in the sense that its solution is always feasible to problem \eqref{originalproblem} but not vice versa due to the inequalities in \eqref{22_app} and \eqref{27_app}, i.e., $R_{\SEC} \geq R_{\SEC}^{\mathtt{LB}}$. In the rest of this paper, we will consider the safe optimization problem \eqref{problemLB} to provide a lower bound of the average secrecy rate  rather than the actual secrecy rate in \eqref{originalproblem}.
\end{remark}

\subsection{Proposed Iterative Algorithm for Solving \eqref{problemLB}}
We are now ready to apply  IA method \cite{Marks:78} to approximate the nonconvex problem \eqref{problemLB}. Before proceeding further, we first introduce new optimization variables $\mathbf{r}\triangleq\{r_{\Ss}[n],r_{\E}[n]\}_{n\in\mathcal{N}}$ to rewrite \eqref{problemLB} equivalently as
\begin{IEEEeqnarray}{rCll}\label{problemLBequi}
	\underline{\mathtt{P}^{\mathtt{Safe}}_{\mathtt{Equi}}:} &\quad \underset{\mathbf{c},\mathbf{p},\mathbf{r}}{\maxtt}&& \quad R_{\SEC}^{\mathtt{LB}}\triangleq \dfrac{1}{N}\sum_{n\in\mathcal{N}}\Bigl(r_{\Ss}[n]- r_{\E}[n]\Bigr)\IEEEyessubnumber\label{problemLBequi:a}\quad\\
	&\quad\st&& \quad\eqref{mobility}, \eqref{originalproblem:c}-\eqref{originalproblem:f}, \eqref{interpowerPR} \IEEEyessubnumber\label{problemLBequi:b},\\
	& && \quad R_{\Ss}^{\mathtt{LB}}[n] \geq r_{\Ss}[n],\ \forall n\in\mathcal{N},  \IEEEyessubnumber\label{problemLBequi:c}\\
	& &&\quad  R_{\E}^{\mathtt{UB}}[n] \leq r_{\E}[n],\ \forall n\in\mathcal{N}, \IEEEyessubnumber\label{problemLBequi:d}
\end{IEEEeqnarray}
It can be readily seen that the objective \eqref{problemLBequi:a} is a linear function of $\mathbf{r}$. In problem \eqref{problemLBequi}, nonconvex parts include  \eqref{interpowerPR}, \eqref{problemLBequi:c} and \eqref{problemLBequi:d}.

\indent\textit{\underline{Convexity of \eqref{problemLBequi:c}:}}
By introducing slack variables $z_{\Ss}\left[n\right]$ and $t_{\Ss}\left[n\right]$,   \eqref{problemLBequi:c} is  expressed as
\begin{subnumcases}{\label{problemLBequi:cequi} \eqref{problemLBequi:c}\Leftrightarrow}
R_{\Ss}^{\mathtt{LB}}[n]\geq \log_2(1 + t_{\Ss}[n])\geq r_{\Ss}[n], & \IEEEyessubnumber\label{problemLBequi:c1}\\
\dfrac{e^{-k}\gamma_{0}d_{\ttSS}^{-\varphi}p_{\Ss}\left[n\right]}{\gamma_{0}\alpha^{-1}_{\Ss}\left[n\right]p_{\UAV}[n]+1} \geq t_{\Ss}[n],& \IEEEyessubnumber\label{problemLBequi:c2}\\
\alpha_{\Ss}\left[n\right]\leq f_{\mathtt{d}}(\mathbf{c}_{\mathtt{S}},\mathbf{c}_{\mathtt{U}}[n])
.&\IEEEyessubnumber\label{problemLBequi:c3}
\end{subnumcases}
We note that constraints \eqref{problemLBequi:c1}-\eqref{problemLBequi:c3} will hold with equality at optimum, leading to an equivalence between \eqref{problemLBequi:c} and \eqref{problemLBequi:cequi}. To avoid the implementation complexity of log function, 
we apply the first-order approximation to approximate the concave function $\log_2(1 + t_{\Ss}[n])$ around the point $t_{\Ss}^{(i)}[n]$ \cite[Eq. (66)]{Nguyen:JSAC:17}, and thus \eqref{problemLBequi:c1} is iteratively approximated as
\begin{equation}\label{problemLBequi:c1:convex}
R_{\Ss}^{(i)}[n] \triangleq a(t_{\Ss}^{(i)}[n])- b(t_{\Ss}^{(i)}[n])\frac{1}{t_{\Ss}[n]}\geq r_{\Ss}[n],\ \forall n\in\mathcal{N},
\end{equation}
where $a(t_{\Ss}^{(i)}[n])\triangleq \log_2(1+t_{\Ss}^{(i)}[n]) + \log_2(e)\frac{t_{\Ss}^{(i)}[n]}{t_{\Ss}^{(i)}[n]+1}$ and $b(t_{\Ss}^{(i)}[n])\triangleq \log_2(e)\frac{(t_{\Ss}^{(i)}[n])^2}{t_{\Ss}^{(i)}[n]+1}$. Next, we rewrite \eqref{problemLBequi:c2} as
\begin{align}\label{problemLBequi:c2equi}
t_{\Ss}\left[n\right]\left(\gamma_{0}p_{\UAV}\left[n\right]+\alpha_\Ss\left[n\right]\right)\leq e^{-k}\gamma_{0}d_{\ttSS}^{-\varphi}p_{\Ss}\left[n\right]\alpha_\Ss\left[n\right],
\end{align}
and then apply the following inequality \cite{NguyenJSAC18}:
\begin{equation}
xy \leq 0.5\Bigr(\frac{y^{(i)}}{x^{(i)}}x^2+\frac{x^{(i)}}{y^{(i)}}y^2\Bigl),\ \text{for} \ x , y \in\mathbb{R}_+, x^{(i)}, y^{(i)} > 0\nonumber,
\end{equation}
to convexify \eqref{problemLBequi:c2equi} as 
\begin{align}\label{problemLBequi:c2:convex}
&\dfrac{1}{2}\dfrac{t_\Ss^{\left(i\right)}\left[n\right]}{\gamma_{0}p_\UAV^{\left(i\right)}\left[n\right]+\alpha_\Ss^{\left(i\right)}\left[n\right]}\left(\gamma_{0}p_\UAV\left[n\right]+\alpha_\Ss\left[n\right]\right)^{2}+  \nonumber\\
&\dfrac{1}{2}\dfrac{\gamma_{0}p_\UAV^{\left(i\right)}\left[n\right]+\alpha_\Ss^{\left(i\right)}\left[n\right]}{t_\Ss^{\left(i\right)}\left[n\right]}t^2_\Ss\left[n\right]+ \dfrac{e^{-k}\gamma_{0}d_{\ttSS}^{-\varphi}}{4}\left(p_{\Ss}\left[n\right]-\alpha_\Ss\left[n\right]\right)^{2}\nonumber\\
&\leq\dfrac{e^{-k}\gamma_{0}d_{\ttSS}^{-\varphi}}{4}\left(\left(p_{\Ss}\left[n\right]+\alpha_\Ss\left[n\right]\right)^{2}\right),\ \forall n\in\mathcal{N}.  
\end{align}
For  constraint \eqref{problemLBequi:c3}, we note that its right-hand side (RHS) is a quadratic convex function which is useful to apply the first-order approximation. Hence, \eqref{problemLBequi:c3} can be iteratively replaced by the following linear constraint:
\begin{align}\label{problemLBequi:c3:convex}
\alpha_\Ss\left[n\right] &\leq  f^{(i)}_{\mathtt{d}}(\mathbf{c}_{\mathtt{U}}[n]|\mathbf{c}_{\mathtt{S}},\mathbf{c}_{\mathtt{U}}^{(i)}[n]),\ \forall n\in\mathcal{N}, 
\end{align}
where $f^{(i)}_{\mathtt{d}}(\mathbf{c}_{\mathtt{U}}[n]|\mathbf{c}_{\mathtt{S}},\mathbf{c}_{\mathtt{U}}^{(i)}[n])$ is the first-order approximation of $f_{\mathtt{d}}(\mathbf{c}_{\mathtt{S}},\mathbf{c}_{\mathtt{U}}[n])$ around the point $\mathbf{c}_{\mathtt{U}}^{(i)}[n]$, which is defined in \eqref{first_order}.
It can be seen that \eqref{problemLBequi:c1:convex}, \eqref{problemLBequi:c2:convex} and \eqref{problemLBequi:c3:convex} are  convex quadratic and linear constraints \cite{Stephen}. 

\indent\textit{\underline{Convexity of \eqref{problemLBequi:d}:}} For new slack variables $t_\E[n], \alpha_\E[n]$ and $\beta[n]$, constraint \eqref{problemLBequi:d} can be rewritten equivalently as
\begin{subnumcases}{\label{problemLBequi:dequi} \eqref{problemLBequi:d}\Leftrightarrow}
R_{\E}^{\mathtt{UB}}[n]\leq \log_2(1 + t_{\E}[n])\leq r_{\E}[n], & \IEEEyessubnumber\label{problemLBequi:d1}\\
\dfrac{\gamma_{0}d_{\SE}^{-\varphi}p_{\Ss}\left[n\right]}{\beta[n]+1} \leq t_{\E}[n],& \IEEEyessubnumber\label{problemLBequi:d2}\\
\beta\left[n\right] \leq \dfrac{\gamma_{0}p_{\UAV}\left[n\right]}{\alpha_\E\left[n\right]},& \IEEEyessubnumber\label{problemLBequi:d3}\\
f_{\mathtt{d}}(\mathbf{c}_{\mathtt{E}},\mathbf{c}_{\mathtt{U}}[n])  
\leq  \alpha_{\E}\left[n\right]. \IEEEyessubnumber\label{problemLBequi:d4}
\end{subnumcases}
In \eqref{problemLBequi:dequi}, except for \eqref{problemLBequi:d4}, other constraints still remain nonconvex.
Since  $\log_2(1 + t_{\E}[n])$ is a concave function, \eqref{problemLBequi:d1} is iteratively replaced by 
\begin{IEEEeqnarray}{rCl}
	R_{\E}^{(i)}[n]&\triangleq& \log_2(1 + t^{(i)}_{\E}[n]) 
	+ \frac{\log_2(e)(t_{\E}[n]-t^{(i)}_{\E}[n])}{1 + t^{(i)}_{\E}[n]}\nonumber\\
	&\leq& r_{\E}[n],\ \forall n\in\mathcal{N},\label{problemLBequi:d1convex}
\end{IEEEeqnarray}
which is a linear constraint.
Similarly to \eqref{problemLBequi:c2:convex},  constraint \eqref{problemLBequi:d2} is approximated around the feasible point $(p^{(i)}_\Ss[n], \beta^{(i)}[n])$ as
\begin{align}\nonumber
\frac{\gamma_{0}d_{\SE}^{-\varphi}}{2}\Bigl( \frac{p^2_{\Ss}[n]}{p^{(i)}_\Ss[n] (\beta^{(i)}[n]+1)}  + \frac{p^{(i)}_\Ss[n] (\beta^{(i)}[n]+1)}{(\beta[n]+1)^2} \Bigr)\leq t_{\E}[n],
\end{align}
which can be cast to the following convex constraint:
\begin{align}\label{problemLBequi:d2convex}
\frac{\gamma_{0}d_{\SE}^{-\varphi}}{2}\Bigl( \frac{p^2_{\Ss}[n]}{p^{(i)}_\Ss[n] (\beta^{(i)}[n]+1)} && \nonumber\\
+ \frac{p^{(i)}_\Ss[n]}{2\beta[n]-\beta^{(i)}[n]+1} \Bigr)\leq&& t_{\E}[n],\ \forall n\in\mathcal{N}.
\end{align}
In \eqref{problemLBequi:d2convex}, the lower bound of $(\beta[n]+1)^2$ is given as $(\beta^{(i)}[n]+1)(2\beta[n]-\beta^{(i)}[n]+1)$ over the trust region $2\beta[n]-\beta^{(i)}[n]+1 > 0$. Constraint \eqref{problemLBequi:d3} is rewritten as $ \alpha_\E\left[n\right]\beta\left[n\right]\leq \gamma_{0}p_{\UAV}\left[n\right]$ and in the same manner as \eqref{problemLBequi:c2:convex}, we have 
\begin{align}\label{problemLBequi:d3convex}
&\dfrac{1}{2}\Bigl(\dfrac{\beta^{\left(i\right)}\left[n\right]}{\alpha_\E^{\left(i\right)}\left[n\right]}\alpha_\E^{2}\left[n\right]+\dfrac{\alpha_\E^{\left(i\right)}\left[n\right]}{\beta^{\left(i\right)}\left[n\right]}\beta^{2}\left[n\right]\Bigr) \leq \gamma_{0}p_{\UAV}\left[n\right],\ \forall n\in\mathcal{N}.
\end{align}

\indent\textit{\underline{Convexity of \eqref{interpowerPR}:}} We first reformulate \eqref{interpowerPR} as
\begin{subnumcases}{\label{interpowerPR:equi} \eqref{interpowerPR}\Leftrightarrow}
\dfrac{1}{N}\sum_{n\in\mathcal{N}}\left(\rho_0 d_{\SP}^{-\varphi}p_{\Ss}[n]+ \rho_0 \frac{p_{\UAV}[n]}{\alpha_\Pp[n]}\right)\leq \varepsilon, & \IEEEyessubnumber\label{interpowerPR:equi:a}\\
\alpha_{\Pp}\left[n\right]\leq f_{\mathtt{d}}(\mathbf{c}_{\mathtt{P}},\mathbf{c}_{\mathtt{U}}[n])\IEEEyessubnumber\label{interpowerPR:equi:b}
\end{subnumcases}
where $\alpha_\Pp[n],\forall n$ are  slack variables. Similarly to \eqref{problemLBequi:d2convex}, constraint \eqref{interpowerPR:equi:a} is iteratively approximated as
\begin{align}\label{interpowerPR:equi:aconvex}
&&\dfrac{1}{N}\sum_{n\in\mathcal{N}}\Biggl(\rho_0 d_{\SP}^{-\varphi}p_{\Ss}[n]+ \frac{\rho_0}{2}\Bigl[\frac{p^2_{\UAV}[n]}{p^{(i)}_{\UAV}[n]\alpha^{(i)}_\Pp[n]}\nonumber\\ 
&&\qquad\qquad\qquad +  \frac{p^{(i)}_{\UAV}[n]}{2\alpha_\Pp[n]-\alpha_\Pp^{(i)}[n]}\Bigr]\Biggr)\leq \varepsilon.
\end{align}
For a given point $ \mathbf{a}=(x_a,y_a,z_a)\in\mathcal{T} $ and  optimization variable $\mathbf{b}=(x_b,y_b,z_b)\in\mathcal{T}$, constraint \eqref{interpowerPR:equi:b} is innerly approximated as
\begin{align}
\alpha_\Pp\left[n\right] &\leq f_{\mathtt{d}}(\mathbf{c}_{\mathtt{P}},\mathbf{c}_{\mathtt{U}}^{(i)}[n]) + f_{\mathtt{g}}(\mathbf{c}_{\mathtt{U}}[n]|\mathbf{c}_{\mathtt{P}},\mathbf{c}_{\mathtt{U}}^{(i)}[n]) \nonumber \\
&\triangleq  f^{(i)}_{\mathtt{d}}(\mathbf{c}_{\mathtt{U}}[n]|\mathbf{c}_{\mathtt{P}},\mathbf{c}_{\mathtt{U}}^{(i)}[n]) \label{32},\ \forall n\in\mathcal{N},
\end{align}
where 
\begin{align} \label{first_order}
f_{\mathtt{g}}\bigl(\mathbf{b}|\mathbf{a},\mathbf{b}^{(i)}\bigr) & \triangleq \left<\nabla f_{\mathtt{d}}(\mathbf{a},\mathbf{b}), \mathbf{b}-\mathbf{b}^{(i)}\right> \nonumber \\
& = \begin{bmatrix}
\nabla_{x} f_{\mathtt{d}}(\mathbf{a},\mathbf{b}) \\
\nabla_{y} f_{\mathtt{d}}(\mathbf{a},\mathbf{b}) \\
\nabla_{z} f_{\mathtt{d}}(\mathbf{a},\mathbf{b})
\end{bmatrix}^T \begin{bmatrix}
x_b-x_b^{(i)} \\
y_b-y_b^{(i)} \\
z_b-z_b^{(i)}
\end{bmatrix} \nonumber \\
& = 2(x_b^{(i)}-x_a)(x_b-x_b^{(i)}) \nonumber \\
&\quad  + 2(y_b^{(i)}-y_a)(y_b-y_b^{(i)}) \nonumber \\
&\quad  + 2(z_b^{(i)}-z_a)(z_b-z_b^{(i)}), 
\end{align}
with $ \nabla_{x} f_{\mathtt{d}}(\mathbf{a},\mathbf{b}) $, $ \nabla_{y} f_{\mathtt{d}}(\mathbf{a},\mathbf{b}) $, and $ \nabla_{z} f_{\mathtt{d}}(\mathbf{a},\mathbf{b}) $ being the  gradients of $ f_{\mathtt{d}}(\mathbf{a},\mathbf{b}) $ with respect to $ x_b $, $ y_b $, and $ z_b $, respectively. In other words, $f^{(i)}_{\mathtt{d}}(\mathbf{c}_{\mathtt{U}}[n]|\mathbf{c}_{\mathtt{P}},\mathbf{c}_{\mathtt{U}}^{(i)}[n])$ is the first-order approximation of $f_{\mathtt{d}}(\mathbf{c}_{\mathtt{P}},\mathbf{c}_{\mathtt{U}}[n])$ around the point $\mathbf{c}_{\mathtt{U}}^{(i)}[n]$.

Bearing all the above developments in mind, the successive convex  program solved at iteration $i$ is given as
\begin{IEEEeqnarray}{rCll}\label{problemLBequi:convex}
	\underline{\mathtt{P}^{\mathtt{Safe}}_{\mathtt{Convex}}:} &\quad \underset{\substack{\mathbf{c},\mathbf{p},\mathbf{r}\\ \mathbf{t},\boldsymbol{\alpha},\boldsymbol{\beta}}}{\maxtt}&& \enspace R_{\SEC}^{\mathtt{LB},(i)}\triangleq \dfrac{1}{N}\sum_{n\in\mathcal{N}}\Bigl(r_{\Ss}[n]- r_{\E}[n]\Bigr)\IEEEyessubnumber\label{problemLBequi:a:convex}\quad\;\\
	&\quad\st&& \ \eqref{mobility}, \eqref{originalproblem:c}-\eqref{originalproblem:f}, \eqref{interpowerPR:equi:aconvex}, \eqref{32}, \eqref{problemLBequi:c1:convex},  \nonumber\\
	& && \ \eqref{problemLBequi:c2:convex}, \eqref{problemLBequi:c3:convex},\eqref{problemLBequi:d4}, \eqref{problemLBequi:d1convex}, \eqref{problemLBequi:d2convex}, \eqref{problemLBequi:d3convex}, \quad \IEEEyessubnumber\label{problemLBequi:b:convex}
\end{IEEEeqnarray}
where $ \mathbf{t}\triangleq\{t_\Ss[n], t_\E[n]\}_{n\in\mathcal{N}}$, $ \boldsymbol{\alpha}\triangleq\{\alpha_\Pp[n], \alpha_\Ss[n], \alpha_\E[n]\}_{n\in\mathcal{N}} $, and $\boldsymbol{\beta}\triangleq\{\beta[n]\}_{n\in\mathcal{N}}$. Let  $\boldsymbol{\Psi}\triangleq\{\mathbf{c},\mathbf{p},\mathbf{r},\mathbf{t},\boldsymbol{\alpha},\boldsymbol{\beta}\}$ and $\boldsymbol{\Psi}^{(i)}\triangleq\{\mathbf{c}^{(i)},\mathbf{p}^{(i)},\mathbf{r}^{(i)},\mathbf{t}^{(i)},\boldsymbol{\alpha}^{(i)},\boldsymbol{\beta}^{(i)}\}$ be the sets of optimization variables and parameters that need to be updated  at iteration $ i $. To ensure that the approximate convex program \eqref{problemLBequi:convex} can be successfully solve at the first iteration, a feasible starting point $\boldsymbol{\Psi}^{(0)}$ must be initialized. We then find the optimal solution of \eqref{problemLB} by successively solving \eqref{problemLBequi:convex} and updating  involved variables until meeting the convergence criterion. In summary, a pseudo-code for solving \eqref{problemLB} is given in Algorithm \ref{proposedAlg1}.


\begin{algorithm}
	\begin{algorithmic}[1]
		\protect\caption{Proposed Algorithm for Solving \eqref{problemLB}}
		\label{proposedAlg1}
		\STATE \textbf{Initialization:} Set $i:=0$ and generate an initial feasible point $\boldsymbol{\Psi}^{(0)}$ satisfying \eqref{problemLBequi:b:convex}. 
		\REPEAT
		\STATE {Set $i:=i+1$}; 
		\STATE {Find the optimal solution $\boldsymbol{\Psi}^{(*)}$ by solving \eqref{problemLBequi:convex}};
		\STATE {Update $\boldsymbol{\Psi}^{(i)}:= \boldsymbol{\Psi}^{(*)}$};
		\UNTIL {$\dfrac{R_{\SEC}^{\mathtt{LB},(i)}-R_{\SEC}^{\mathtt{LB},(i-1)}}{R_{\SEC}^{\mathtt{LB},(i-1)}} \leq \epsilon_{\mathtt{tol}}$}.		
	\end{algorithmic} 
\end{algorithm}

{\noindent \textit{\underline{Complexity Analysis:}} The optimization problem \eqref{problemLBequi:convex} has $13N$ real variables and $16N$ constraints. The per-iteration complexity of Algorithm 1 required
to solve \eqref{problemLBequi:convex} is thus $\mathcal{O}((16N)^{2.5}(13N)^{2}+(16N)^{3.5}).$}

\section{Extension to the Case of Imperfect Location Information of Eavesdropper}

\begin{figure}[t]
	\begin{center}
		\includegraphics[width=0.8\columnwidth] {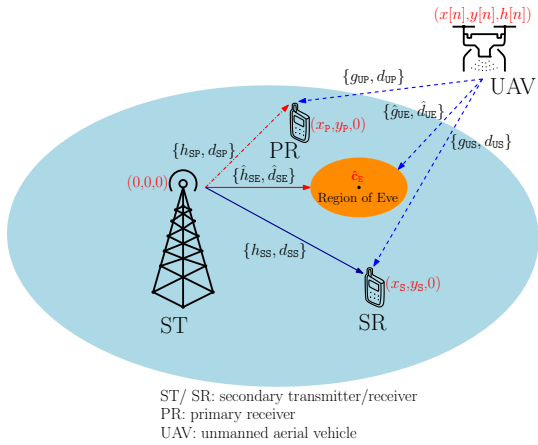}
	\end{center}
	\caption{Illustration of a CRN with a UAV-aided JN and an external Eve under  imperfect location information of Eve.\label{fig:im_systemmodel}}
\end{figure}
In practice, perfect information on the location of Eve may be difficult to obtain in some cases. For example, Eve can move
to new locations (e.g., closer to the ST) to overhear confidential messages from ST more effectively. As a result, the location 
of Eve may change, and thus, it can only be estimated by ST and UAV  based on its last known
location. Moreover, the active region of Eve may be restricted, and Eve may not be allowed to move inside the safe zone of the ST. 

 As illustrated in Fig. \ref{fig:im_systemmodel}, the location estimate for Eve can be expressed as $\mathbf{\hat{c}}_{\mathtt{E}}=(\hat{x}_{\mathtt{E}},\hat{y}_{\mathtt{E}},\hat{z}_{\mathtt{E}}=0)$. We consider the same optimization problem setup as in Section III, with
the additional assumption that the location information of Eve is imperfect. To put it into context, let
\begin{IEEEeqnarray}{rCl}\label{imperfect}
	x_{\mathtt{E}} &=& \hat{x}_{\mathtt{E}} + \Delta x_{\mathtt{E}},\IEEEyessubnumber\label{17}\\
	y_{\mathtt{E}} &=& \hat{y}_{\mathtt{E}} + \Delta y_{\mathtt{E}},\IEEEyessubnumber\label{18}\\
	z_{\mathtt{E}} &=& \hat{z}_{\mathtt{E}}=0,\IEEEyessubnumber\label{18-z}
\end{IEEEeqnarray} 
where $\Delta x_{\mathtt{E}}$ and $\Delta y_{\mathtt{E}}$ represent the associated estimation errors of $x_{\mathtt{E}}$ and $y_{\mathtt{E}}$, respectively. We should note that  the transmitters are only aware of $\mathbf{\hat{c}}_{\mathtt{E}}$, while the estimation errors $\Delta x_{\mathtt{E}}$ and $\Delta y_{\mathtt{E}}$ are assumed to be deterministic and bounded, satisfying the following condition \cite{ImperfectLocation, Q_Li}:
\begin{IEEEeqnarray}{rCl}
	(\Delta x_{\mathtt{E}},\Delta y_{\mathtt{E}}) \in \Xi \triangleq \{(\Delta x_{\mathtt{E}},\Delta y_{\mathtt{E}})|\Delta x_{\mathtt{E}}^{2} + \Delta y_{\mathtt{E}}^{2} \leq Q^{2}\}, \label{19}\quad
\end{IEEEeqnarray}
where $Q>0$ is the maximum distance between the estimate and exact location of Eve.

\subsection{Worst-Case Optimization Problem Formulation}
Toward a safe design, the worst-case secrecy rate is considered. We first introduce the following lemma:
\begin{lemma} \label{lemma 4}
	Consider that the location estimation error of Eve  is  deterministic and bounded: $ (\Delta x_{\E},\Delta y_{\E}) \in \Xi $. By utilizing the tractable form in Section \ref{Alg Under Perfect}-B, we formulate the worst-case secrecy rate of the secondary system \cite{Q_Li} at time slot $n$ as 
	\begin{equation}\label{worst-case}
		\bar{R}_{\mathtt{sec}}[n] = R_{\mathtt{S}}^{\mathtt{LB}}\left[n\right] - \underset{(\Delta x_{\E},\Delta y_{\E}) \in \Xi}{\maxtt} \enspace R_{\mathtt{E}}^{\mathtt{UB}}(\hat{d}_{\mathtt{SE}},\hat{d}_{\mathtt{UE}}\left[n\right]),
	\end{equation}
	where $ R_{\mathtt{S}}^{\mathtt{LB}}\left[n\right] $ is given in \eqref{22_app}; $ R_{\mathtt{E}}^{\mathtt{UB}}(\hat{d}_{\mathtt{SE}},\hat{d}_{\mathtt{UE}}\left[n\right]) $ is  a function of $ (\hat{d}_{\mathtt{SE}},\hat{d}_{\mathtt{UE}}\left[n\right]) $, instead of $ (d_{\mathtt{SE}},d_{\mathtt{UE}}\left[n\right]) $ in \eqref{27_app}. For a tractable form, the worst-case secrecy rate in \eqref{worst-case} is further transformed  into a ``strict''  worst-case secrecy rate:
	\begin{align} \label{strict worst-case}
		\hat{R}_{\mathtt{sec}}[n] = R_{\mathtt{S}}^{\mathtt{LB}}\left[n\right] - \underset{(\Delta x_{\E},\Delta y_{\E}) \in \Xi}{\maxtt} \enspace \hat{R}_{\mathtt{E}}^{\mathtt{UB}}\left[n\right],
	\end{align}
	where \[ \hat{R}_{\mathtt{E}}^{\mathtt{UB}}\left[n\right]\triangleq R_{\mathtt{E}}^{\mathtt{UB}}[n]\bigl(\ddot{d}_{\mathtt{SE}},\hat{d}_{\mathtt{UE}}[n]\bigr)= \sup_{\hat{d}_{\mathtt{SE}} \in \mathcal{D}} R_{\mathtt{E}}^{\mathtt{UB}}[n]\bigl(\hat{d}_{\mathtt{SE}},\hat{d}_{\mathtt{UE}}[n]\bigr), \] with $\mathcal{D}$ being the set of  distances from ST to Eve. As illustrated in Fig. \ref{fig:intersection}, a fixed distance $ \ddot{d}_{\mathtt{SE}} $ is determined by $ \ddot{d}_{\mathtt{SE}}=f_{\mathtt{d}}(\mathbf{\ddot{c}}_{\mathtt{E}},\mathbf{c}_{\mathtt{ST}}) $, where $ \mathbf{\ddot{c}}_{\mathtt{E}} $ is  the nearest geometric point such that $ \mathbf{\ddot{c}}_{\mathtt{E}}\in\{\mathbf{\hat{c}}_{\mathtt{E}}+(\Delta x_{\E}, \Delta y_{\E}, 0)|(\Delta x_{\E}, \Delta y_{\E})\in\Xi\} $.
\end{lemma}

\begin{figure}[t]
	\begin{center}
		\includegraphics[width=1\columnwidth] {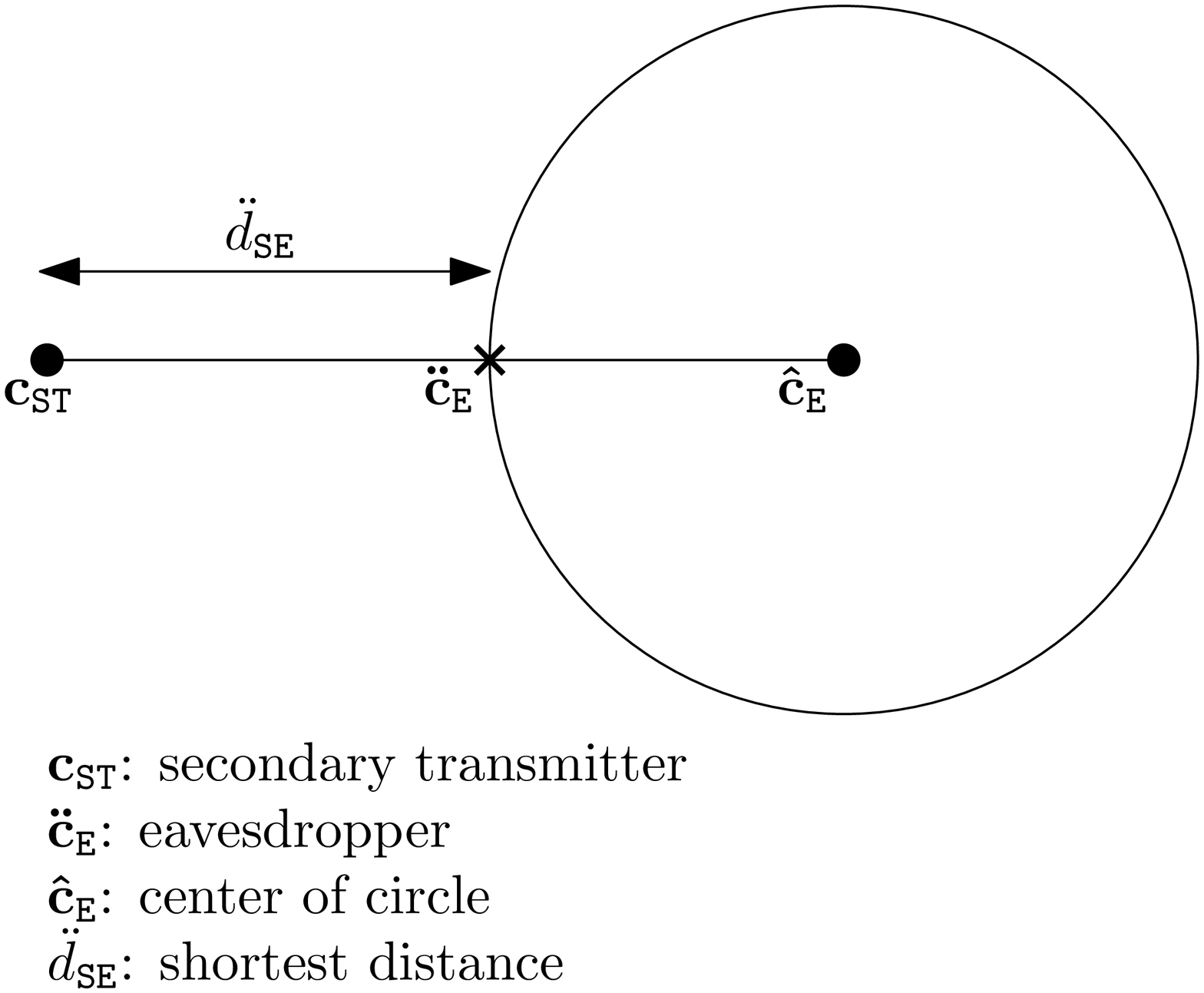}
	\end{center}
	\caption{The possible location of Eve in the ``strict" worst-case optimization problem.\label{fig:intersection}}
\end{figure}

\begin{proof}
	Please see Appendix \ref{app D}.
\end{proof}

It can be foreseen that this analysis can further reduce the complexity of the optimization problem, since $\hat{d}_{\mathtt{SE}}$ is replaced by $\ddot{d}_{\mathtt{SE}}$. Nevertheless, the property of the worst-case secrecy rate over the set of $(\Delta x,\Delta y)$ would be strictly remained when addressing $\hat{d}_{\mathtt{UE}}[n]$. Based on the developments presented in Section III-B, we formulate the strict average worst-case SRM (WC-SRM) problem of  CRN as follows:
\begin{IEEEeqnarray}{lCll}\label{im_originalproblem}
	\underline{\mathtt{\hat{P}}^{\mathtt{Safe}}:}& \quad \underset{\mathbf{c},\mathbf{p}}{\maxtt} &&\quad \hat{R}_{\SEC}^{\mathtt{LB}}\triangleq\frac{1}{N}\sum_{n\in\mathcal{N}}\hat{R}_{\mathtt{sec}}[n]\IEEEyessubnumber\label{im_originalproblem:a}\\
	&\quad\st&& \quad\eqref{mobility}, \eqref{originalproblem:c}-\eqref{originalproblem:f}, \eqref{interpowerPR} \IEEEyessubnumber\label{im_originalproblem:b}.
\end{IEEEeqnarray}
It can be seen that \eqref{problemLBequi} and \eqref{im_originalproblem} have similar structure and the same set of constraints. However, the objective function of \eqref{im_originalproblem} is more complex due to joint optimization under estimation errors, making the problem even more challenging to solve.

\subsection{Proposed Iterative Algorithm for Solving \eqref{im_originalproblem}}
In this section, we reuse all the slack optimization variables introduced in Section III. By following the same steps presented in Section III-C,
we arrive at the following safe and approximate optimization problem for the WC-SRM  \eqref{im_originalproblem}:  
\begin{IEEEeqnarray}{rCll}\label{im_problemLBequi}
	\underline{\mathtt{\hat{P}}^{\mathtt{Safe}}_{\mathtt{Appr}}:} & \underset{\substack{\mathbf{c},\mathbf{p},\mathbf{r} \\ \mathbf{t}_{\mathtt{S}}, \boldsymbol{\alpha}_{\mathtt{PS}} }}{\maxtt}&& \enspace \hat{R}_{\SEC}^{\mathtt{LB}}\triangleq \dfrac{1}{N}\sum_{n\in\mathcal{N}}\Bigl(r_{\mathtt{S}}[n]- r_{\mathtt{E}}[n]\Bigr) \IEEEyessubnumber\label{im_problemLBequi:a}\\
	\st& && \eqref{mobility}, \eqref{originalproblem:c}-\eqref{originalproblem:f},  \eqref{problemLBequi:c1:convex}, \eqref{problemLBequi:c2:convex}, \eqref{problemLBequi:c3:convex}, \eqref{interpowerPR:equi:aconvex}, \eqref{32},\qquad\IEEEyessubnumber\label{im_problemLBequi:b}\\
	& &&  \underset{(\Delta x_{\E},\Delta y_{\E}) \in \Xi}{\maxtt}  \hat{R}_{\mathtt{E}}^{\mathtt{UB}}\left[n\right] \leq r_{\mathtt{E}}[n], \ \forall n\in\mathcal{N}, \qquad \IEEEyessubnumber \label{im_problemLBequi:d}
\end{IEEEeqnarray}
where $ \mathbf{t}_{\Ss}\triangleq\{t_\Ss[n]\}_{n\in\mathcal{N}}$ and $ \boldsymbol{\alpha}_{\mathtt{PS}}\triangleq\{\alpha_\Pp[n],\alpha_\Ss[n]\}_{n\in\mathcal{N}} $. In \eqref{im_problemLBequi}, the objective is already linear function while all the constraints are convex, excepting for \eqref{im_problemLBequi:d}.

\indent\textit{\underline{Convexity of \eqref{im_problemLBequi:d}:}}
Similarly to \eqref{problemLBequi:dequi}, it follows that
\begin{subnumcases}{\label{im_problemLBequi:dequi} \eqref{im_problemLBequi:d}\Leftrightarrow}
\hat{R}_{\mathtt{E}}^{\mathtt{UB}}[n]\leq \log_2(1 + t_{\mathtt{E}}[n])\leq r_{\mathtt{E}}[n], \IEEEyessubnumber\label{im_problemLBequi:d1}\\
\dfrac{\gamma_{0}\ddot{d}_{\mathtt{SE}}^{-\varphi}p_{\mathtt{S}}\left[n\right]}{\beta[n]+1} \leq t_{\mathtt{E}}[n], \IEEEyessubnumber\label{im_problemLBequi:d2}\\
\beta\left[n\right] \leq \dfrac{\gamma_{0}p_{\mathtt{U}}\left[n\right]}{\alpha_{\mathtt{E}}\left[n\right]}, \IEEEyessubnumber\label{im_problemLBequi:d3}\\
\underset{(\Delta x_{\E},\Delta y_{\E}) \in \Xi}{\maxtt} f_{\mathtt{d}}(\mathbf{\hat{c}}_{\mathtt{E}}+(\Delta x_{\E}, \Delta y_{\E}, 0),\mathbf{c}_{\mathtt{U}}[n]) \qquad\nonumber\\
 \leq  \alpha_{\E}\left[n\right]. \qquad \IEEEyessubnumber\label{im_problemLBequi:d4} 
\end{subnumcases}
Constraints \eqref{im_problemLBequi:d1} and \eqref{im_problemLBequi:d3} are tackled as the same steps in \eqref{problemLBequi:d1} and \eqref{problemLBequi:d3}, respectively; \eqref{im_problemLBequi:d2}  can be convexified by replacing $ d_{\mathtt{SE}} $ in \eqref{problemLBequi:d2} with $\ddot{d}_{\mathtt{SE}}$ as
\begin{align}\label{im_problemLBequi:d2convex}
\frac{\gamma_{0}\ddot{d}_{\mathtt{SE}}^{-\varphi}}{2}\Bigl( \frac{p^2_{\mathtt{S}}[n]}{p^{(i)}_{\mathtt{S}}[n] (\beta^{(i)}[n]+1)}&&  \nonumber\\
  + \frac{p^{(i)}_{\mathtt{S}}[n]}{2\beta[n]-\beta^{(i)}[n]+1} \Bigr)&\leq& t_{\mathtt{E}}[n],\ \forall n\in\mathcal{N}.
\end{align}
Since $\Xi$ is a continuous set of estimation errors, enumerating all the possible cases of $(\Delta x,\Delta y)$ is obviously impossible.  To overcome this issue, we first reformulate \eqref{im_problemLBequi:d4} as follows:
\begin{subnumcases}{\label{estimation_error} \eqref{im_problemLBequi:d4}\Leftrightarrow}
\Delta x_{\E}^{2} + \Delta y_{\E}^{2} \leq Q^{2}, & \IEEEyessubnumber\label{estimation_error:a}\\
f_{\mathtt{d}}(\mathbf{\hat{c}}_{\mathtt{E}}+(\Delta x_{\E}, \Delta y_{\E}, 0),\mathbf{c}_{\mathtt{U}}[n])  \leq \alpha_{\E}\left[n\right]. \IEEEyessubnumber \label{estimation_error:b}
\end{subnumcases}
To address the nonconvex constraint \eqref{estimation_error}, we  introduce the following lemma.
\begin{lemma} \label{lemma 5}
	By applying S-procedure and Schur's complement \cite{Stephen}, \eqref{estimation_error} is transformed into the following convex constraints:
	\begin{subequations} \label{R_E UB convex}
		\begin{align}
			f_{\mathtt{d}}(\mathbf{\hat{c}}_{\mathtt{E}},\mathbf{c}_{\mathtt{U}}[n]) - \alpha_{\mathtt{E}}[n] & \leq \theta_{\E}[n], \ \forall n\in\mathcal{N}, \\
			\mu[n] & \geq 0, \ \forall n\in\mathcal{N}, \\
			\mathbf{S}[n] &\succeq 
			\mathbf{0}, \ \forall n\in\mathcal{N},
		\end{align}
	\end{subequations}
	where $\boldsymbol{\theta}\triangleq \{\theta_{\E}[n]\}_{n\in\mathcal{N}}$ and $\boldsymbol{\mu}\triangleq \{\mu[n]\}_{n\in\mathcal{N}}$ are  slack variables, and
	\begin{align}
		\mathbf{S}[n]\triangleq
		\begin{bmatrix}
		\mu[n]-1 & 0 & x_{\mathtt{U}}[n] - \hat{x}_{\mathtt{E}} \\
		0 & \mu[n]-1 & y_{\mathtt{U}}[n] - \hat{y}_{\mathtt{E}} \\
		x_{\mathtt{U}}[n] - \hat{x}_{\mathtt{E}} & y_{\mathtt{U}}[n] - \hat{y}_{\mathtt{E}} & -Q^{2}\mu[n]-\theta_{\E}[n]
		\end{bmatrix}. \nonumber
	\end{align}
\end{lemma} 

\begin{proof}
	Please see Appendix \ref{app E}. 
\end{proof}

As summarized in Algorithm \ref{proposedAlg}, the solution of the WC-SRM problem \eqref{im_originalproblem} can be found by successively solving a safe and convex program, of which the approximated problem at iteration $i+1$ is expressed as
\begin{IEEEeqnarray}{rCll}\label{im_problemLBequi:convex}
	\underline{\mathtt{\hat{P}}^{\mathtt{Safe}}_{\mathtt{Convex}}:} &\quad \underset{\boldsymbol{\hat{\Psi}}}{\maxtt}&& \quad \hat{R}_{\mathtt{sec}}^{\mathtt{LB},(i)}\triangleq \dfrac{1}{N}\sum_{n\in\mathcal{N}}\Bigl(r_{\mathtt{S}}[n]- r_{\mathtt{E}}[n]\Bigr) \qquad \IEEEyessubnumber\label{im_problemLBequi:a:convex} \\
	&\quad\st&& \quad \eqref{mobility}, \eqref{originalproblem:c}-\eqref{originalproblem:f}, \eqref{problemLBequi:c1:convex}, \eqref{problemLBequi:c2:convex}, \eqref{problemLBequi:c3:convex}, \nonumber \\
	& && \quad  \eqref{problemLBequi:d1convex}, \eqref{problemLBequi:d3convex}, \eqref{interpowerPR:equi:aconvex}, \eqref{32}, \eqref{im_problemLBequi:d2convex}, \eqref{R_E UB convex}, \qquad \IEEEyessubnumber\label{im_problemLBequi:b:convex}
\end{IEEEeqnarray}
where $\boldsymbol{\hat{\Psi}}\triangleq \{\mathbf{c},\mathbf{p},\mathbf{r},\mathbf{t},\boldsymbol{\alpha},\boldsymbol{\beta},\boldsymbol{\theta},\boldsymbol{\mu}\}$, which correspondingly provides  $\boldsymbol{\hat{\Psi}}^{(i)}\triangleq\{\mathbf{c}^{(i)},\mathbf{p}^{(i)},\mathbf{r}^{(i)},\mathbf{t}^{(i)},\boldsymbol{\alpha}^{(i)}, $ $\boldsymbol{\beta}^{(i)},\boldsymbol{\theta}^{(i)},\boldsymbol{\mu}^{(i)}\}$ as the optimal solution for \eqref{im_problemLBequi:convex} at iteration $i $. 
\begin{algorithm}
	\begin{algorithmic}[1]
		\protect\caption{Proposed Algorithm for Solving \eqref{im_originalproblem}}
		\label{proposedAlg}
		\STATE \textbf{Initialization:} Set $i:=0$ and generate an initial feasible point $\boldsymbol{\hat{\Psi}}^{(0)}$ satisfying \eqref{im_problemLBequi:b:convex}. 
		\REPEAT
		\STATE {Set $i:=i+1$}; 
		\STATE {Find the optimal solution $\boldsymbol{\hat{\Psi}}^{(*)}$ by solving \eqref{im_problemLBequi:convex}};
		\STATE Update $\boldsymbol{\hat{\Psi}}^{(i)}=\boldsymbol{\hat{\Psi}}^{(i-1)}$;
		\UNTIL {$\dfrac{\hat{R}_{\mathtt{sec}}^{\mathtt{LB},(i)}-\hat{R}_{\mathtt{sec}}^{\mathtt{LB},(i-1)}}{\hat{R}_{\mathtt{sec}}^{\mathtt{LB},(i-1)}} \leq \epsilon_{\mathtt{tol}}$}.		
	\end{algorithmic} 
\end{algorithm}

{\noindent \textit{\underline{Complexity Analysis:}} The optimization problem \eqref{im_problemLBequi:convex} has $15N$ real variables and $18N$ constraints}. The complexity required
to solve \eqref{im_problemLBequi:convex} in each iteration of Algorithm 2  is  $\mathcal{O}((18N)^{2.5}(15N)^{2}+(18N)^{3.5})$.


\subsection{Convergence Analysis of Algorithms 1 and 2}  
We can  see that the objective values in  \eqref{problemLBequi:a:convex} and \eqref{im_problemLBequi:a:convex} are non-decreasing with respect to the number of iterations, and the convergence proof for the optimization problems is given in \cite[Appendix C]{NguyenJSAC18}. To be
self-contained, we briefly provide the convergence analysis as follows. We can  see that the approximations of nonconvex constraints \{\eqref{interpowerPR}, \eqref{problemLBequi:c},  \eqref{problemLBequi:d}\} for problem \eqref{problemLB} and \{\eqref{interpowerPR}, \eqref{problemLBequi:c}, \eqref{im_problemLBequi:d}\} for problem \eqref{im_originalproblem} satisfy  properties of the IA method given in \cite{Marks:78}. This  means that the proposed Algorithms 1 and 2 for solving \eqref{problemLBequi:convex} and \eqref{im_problemLBequi:convex}, respectively, generate the sequences of non-decreasing objective values (i.e., $R_{\SEC}^{\mathtt{LB},(i)}\geq R_{\SEC}^{\mathtt{LB},(i-1)}$ and $\hat{R}_{\mathtt{sec}}^{\mathtt{LB},(i)}\geq \hat{R}_{\mathtt{sec}}^{\mathtt{LB},(i-1)}$), which are upper bounded  due to the power constraints, leading to a monotonic convergence. At each
iteration, the achieved optimal solutions satisfy the Karush-Kuhn-Tucker (KKT) conditions of \eqref{problemLBequi:convex} and \eqref{im_problemLBequi:convex}, i.e., step 4 of Algorithms 1 and 2, respectively. By IA principle, the KKT conditions of \eqref{problemLBequi:convex} and \eqref{im_problemLBequi:convex} are also identical to those of \eqref{problemLB} and \eqref{im_originalproblem}, respectively, once the conditions $\boldsymbol{\Psi}^{(i)}=\boldsymbol{\Psi}^{(i-1)}$ (in Algorithm 1) and $\boldsymbol{\hat{\Psi}}^{(i)}=\boldsymbol{\hat{\Psi}}^{(i-1)}$ (in Algorithm 2) are met \cite[Theorem 1]{Marks:78}.


\section{Numerical Results}
\begin{table}[t]
\caption{Simulation Parameters}
	\label{parameter}
	\centering
		\begin{tabular}{l|l}
		\hline
				Parameter & Value \\
		\hline\hline
		    System bandwidth                            & 10 MHz \\
				Path loss exponent, $\varphi$  & 3\\
				Number of time slots, $N$ & 500\\
				Channel gain at the reference distance, $\rho_0$& 10 dB\\
				Power budget at  ST, $P_{\Ss}^{\max}$ & 40 dBm\\
				Average power limit at  ST, $\bar{P}_{\Ss}$ & $P_{\Ss}^{\max}/2$\\
				Power budget at  UAV, $P_{\UAV}^{\max}$ & 4 dBm\\
				Average power limit at  UAV, $\bar{P}_{\UAV}$ & $P_{\UAV}^{\max}$/2\\
				Maximum and minimum altitudes of  UAV, $(h^{\max},h^{\min})$ & (150,\ 50) m\\
				Maximum speed of  UAV, $V_{\max}$& 10 m/s\\
				Average interference power threshold at  PR, $\varepsilon$ & -20 dBm\\
				Noise power, $\sigma^2$               & -70 dBm\\
				Error tolerance threshold, $\epsilon_{\mathtt{tol}}$ & $10^{-4}$\\
					\hline		   				
		\end{tabular}
\end{table}
We now evaluate the performance of the proposed schemes using computer simulations in the MATLAB environment. The key parameters are given in Table \ref{parameter}. 
The ST, SR and PR are assumed to locate at $\left(0,0,0\right)$, $\left(300,0,0\right)$ and $\left(0,250,0\right)$, respectively. We also assume that UAV flies from the original location at $(-100,200,100)$ to the destination at $(500,200,100)$. The other parameters are provided in the captions of the figures. 
The  convex solver SeDuMi is used to solve the convex program.

The results obtained by Algorithms 1 and 2 are labeled as ``Proposed scheme (Alg. 1)" and ``Proposed scheme (Alg. 2)", respectively. For comparison purpose, we investigate three other schemes:
\begin{itemize}
\item ``Fixed power:'' In every time slot, ST and UAV transmit their signals with the fixed transmit powers, i.e., $\bar{P}_{\Ss}$ and $\bar{P}_{\UAV}$, respectively, and only the UAV's trajectory is optimized. 
\item ``Straight line trajectory:'' The UAV flies along the straight line from the initial location to the final location, and only the transmit power of ST and UAV is optimized.
\item ``No UAV-aided JN:''  We set $p_{\UAV}[n] = 0, \forall n$ (i.e., without using UAV-aided JN), which corresponds to the traditional on-ground CRN.
\end{itemize}
The solutions of these schemes can  also be obtained by using Algorithms 1 and 2 after some slight modifications.

\subsection{Numerical Results for Perfect Location Information of Eve}
In this scenario,  Eve is placed at (150, 250, 0), which is closer to ST than SR. This unfair setting aims at demonstrating the effectiveness of using UAV-aided JN.
\begin{figure}
    \begin{center}
    \includegraphics[width=0.8\columnwidth,trim={0cm 0cm 0cm 0}]{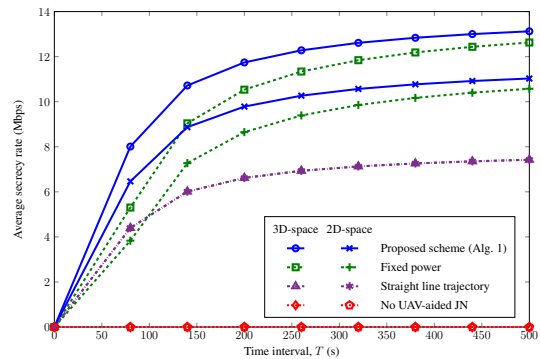}
    \end{center}
    \caption{Average secrecy rate versus time interval $T$, with perfect location information of Eve.\label{fig:SRvsTime}}
\end{figure}

In Fig. \ref{fig:SRvsTime}, the average secrecy rates of different schemes are illustrated versus the time interval, $T\in[0,500\text{s}]$. It is not difficult to see that the average secrecy rate is always less than or equal to zero in the case of ``No UAV-aided JN''  scheme. The reason is that the ST-SR link has worse channel quality than the ST-Eve link. This result verifies the importance of using UAV-aided JN. The other important observations from the figure are as follows. First, all schemes provide the non-decreasing secrecy rates as $T$ increases. This is because the larger $T$ the larger time for UAV to hover over Eve to transmit JN more effectively. Second, from numerical results of the average secrecy rates of ``Straight line trajectory'' when compared to the ``Proposed method (Alg. 1)'' and ``Fixed power'', we can see that the UAV's trajectory optimization is highly important, since it can help UAV fly to an optimal location to interfere with the ST-Eve channel. Third, the proposed method always provides the best performance along with $ T $. Finally, the secrecy rate of the proposed scheme in the 3D space is  superior to that in the 2D space, and
an improvement of almost $2$ Mbps is achieved at $T = 400$s.   

\begin{figure}[t]
	\centering
	\begin{subfigure}[Trajectories of UAV in the 2D space.]
		{
			\includegraphics[width=0.8\columnwidth,trim={0cm 0cm 0cm 0cm}]{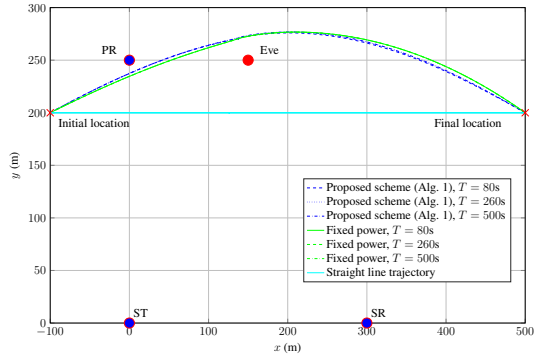}
			\label{fig:Trajectory}
		}
	\end{subfigure}
	\hfill
	\begin{subfigure}[Trajectories of UAV in the 3D space.]
		{
			\includegraphics[width=0.8\columnwidth,trim={0cm 0cm 0cm 0cm}]{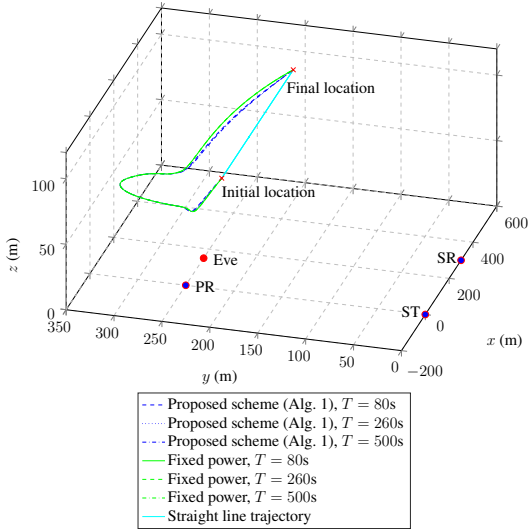}
			\label{fig:Trajectory-3D}
		}
	\end{subfigure}
	\caption{Trajectories of UAV for different schemes with perfect location information of Eve.}
	\label{fig:Trajectory-2D3D}
\end{figure}


The trajectories of UAV are depicted in Figs. \ref{fig:Trajectory} and \ref{fig:Trajectory-3D} for different schemes with $T\in\{80\text{s}, 260\text{s}, 500\text{s}\}$ in both the 2D and 3D spaces. Except for ``Straight line trajectory", the other schemes follow similar trajectories,  since  UAV aims at emitting  JN to jam Eve in a short distance (but keep far away from  SR to mitigate the interference caused by JN), as long as satisfying the PR's interference power requirement. Furthermore, the distances between  Eve and UAV are defined as a function of $n$ in Fig. \ref{fig:DisUAV-Eve}. Although the optimal UAV-Eve distance is intuitively $ 100 $ m, UAV does not move to the point above Eve directly. To maximize the average secrecy rate, the UAV trajectory is optimized under a tradeoff between the secrecy performance improvement and the amount of undesired interference  to  SR and PR.



\begin{figure}
    \begin{center}
    \includegraphics[width=1\columnwidth,trim={0cm 0cm 0cm 0cm}]{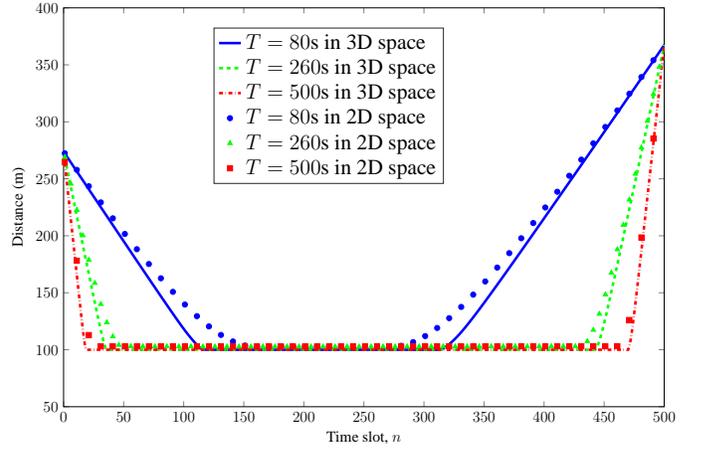}
    \end{center}
    \caption{UAV-Eve distance per time slot $n$ during time interval $T$, with perfect location information of Eve.\label{fig:DisUAV-Eve}}
\end{figure}

 


\begin{figure}
	\begin{center}
		\includegraphics[width=0.8\columnwidth,trim={0cm 0cm 0cm 0cm}]{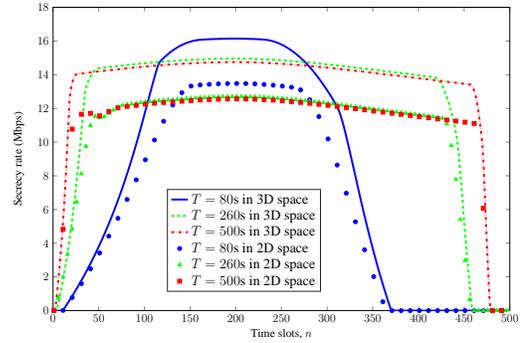}
	\end{center}
	\caption{Secrecy rate of Algorithm 1 per time slot $n$ during time interval $T$, with perfect location information of Eve.\label{fig:SecrecyRate_Timeslot}}
\end{figure}
Fig. \ref{fig:SecrecyRate_Timeslot} depicts the secrecy rate of Algorithm 1 per time slot with different values of $T$ in both the 2D and 3D spaces. One can see that the number of time slots having the positive secrecy rate in the 3D space is much higher than that in the 2D space, which demonstrates the effectiveness of jointly optimizing the UAV's altitude. This phenomenon can be further confirmed by the results in Fig. \ref{fig:DisUAV-Eve}, where the number of time slots having the optimal UAV-Eve distance in the 3D space is higher than that in the 2D space. Moreover, the secrecy rates reduce to zero at the last time slots. This is because UAV moves closer to SR than Eve at those time slots, and thus, it must stop sending JN.


\subsection{Numerical Results for Imperfect Location Information of Eve}
We assume that Eve is located in a circular region  centered at $(x_{\mathtt{E}_{0}},y_{\mathtt{E}_{0}},h_{\mathtt{E}_{0}})=(150,250,0)$ with the radius $Q=20$ m. The other simulation parameters are the same as before. 

\begin{figure}[t]
	\centering
	\begin{subfigure}[Average secrecy rate versus the time interval $T$.]
		{
			\includegraphics[width=0.8\columnwidth,trim={0cm 0cm 0cm 0cm}]{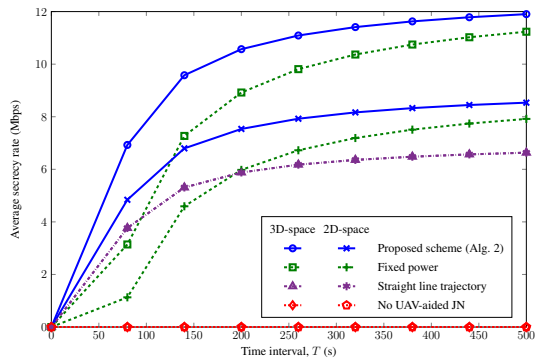}
			\label{fig:Im_SRvsTime}
		}
	\end{subfigure}
	\hfill
	\begin{subfigure}[Secrecy rate of Algorithm 2 per time slots $n$ during time interval $T$.]
		{
			\includegraphics[width=0.8\columnwidth,trim={0cm 0cm 0cm 0cm}]{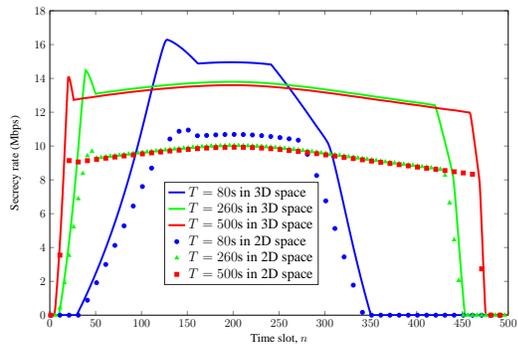}
			\label{fig:Im_SecrecyRate_Timeslot}
		}
	\end{subfigure}
	\caption{Secrecy rates with imperfect location information of Eve.}
	\label{fig:Im_SRvsTN}
\end{figure}

We plot the average secrecy rate versus the time interval $T$ in Fig. \ref{fig:Im_SRvsTime} and the secrecy rate of Algorithm 2 per time slot with different values of $T$ in Fig. \ref{fig:Im_SecrecyRate_Timeslot}. Unsurprisingly, the secrecy rate of all schemes is degraded, when compared to the case of perfect location information of Eve. Notably, the performance gaps between 3D and 2D cases are even deeper. In Fig. \ref{fig:Im_SRvsTime}, at $T=400$s, the performance gain of 3D over 2D is about 3.5 Mbps, compared to 2 Mbps in Fig. \ref{fig:SRvsTime}. These results confirm the robustness of the proposed scheme
against the  effect of imperfect location information of Eve.  Fig. \ref{fig:Im_Trajectory-2D3D} illustrates the trajectories of UAV in the 2D and 3D spaces, and we recall the discussions presented for Fig. \ref{fig:Trajectory-2D3D}.

\begin{figure}[t]
	\centering
	\begin{subfigure}[Trajectories of UAV in the 2D space.]
		{
			\includegraphics[width=0.9\columnwidth,trim={0cm 0cm 0cm 0cm}]{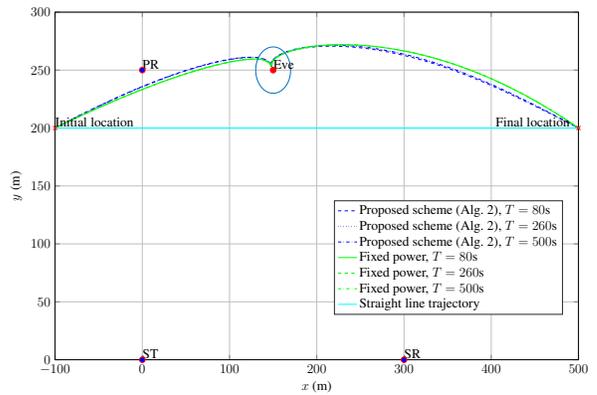}
			\label{fig:Im_Trajectory}
		}
	\end{subfigure}
	\hfill
	\begin{subfigure}[Trajectories of UAV in the 3D space.]
		{
			\includegraphics[width=1.2\columnwidth,trim={0cm 0cm 0cm 0cm}]{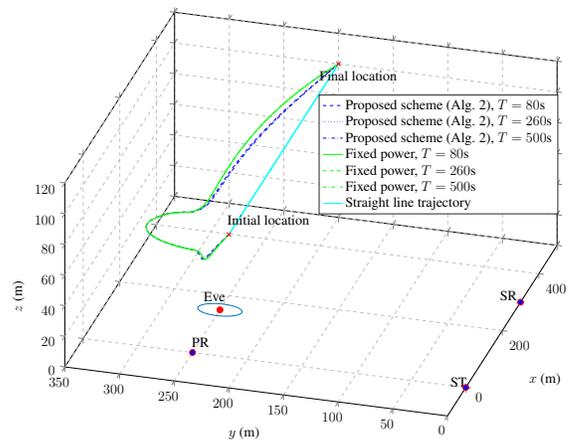}
			\label{fig:Im_Trajectory-3D}
		}
	\end{subfigure}
	\caption{Trajectories of UAV for different schemes with imperfect location information of Eve.}
	\label{fig:Im_Trajectory-2D3D}
\end{figure}

\begin{figure}
	\begin{center}
		\includegraphics[width=0.8\columnwidth,trim={0cm 0cm 0cm 0}]{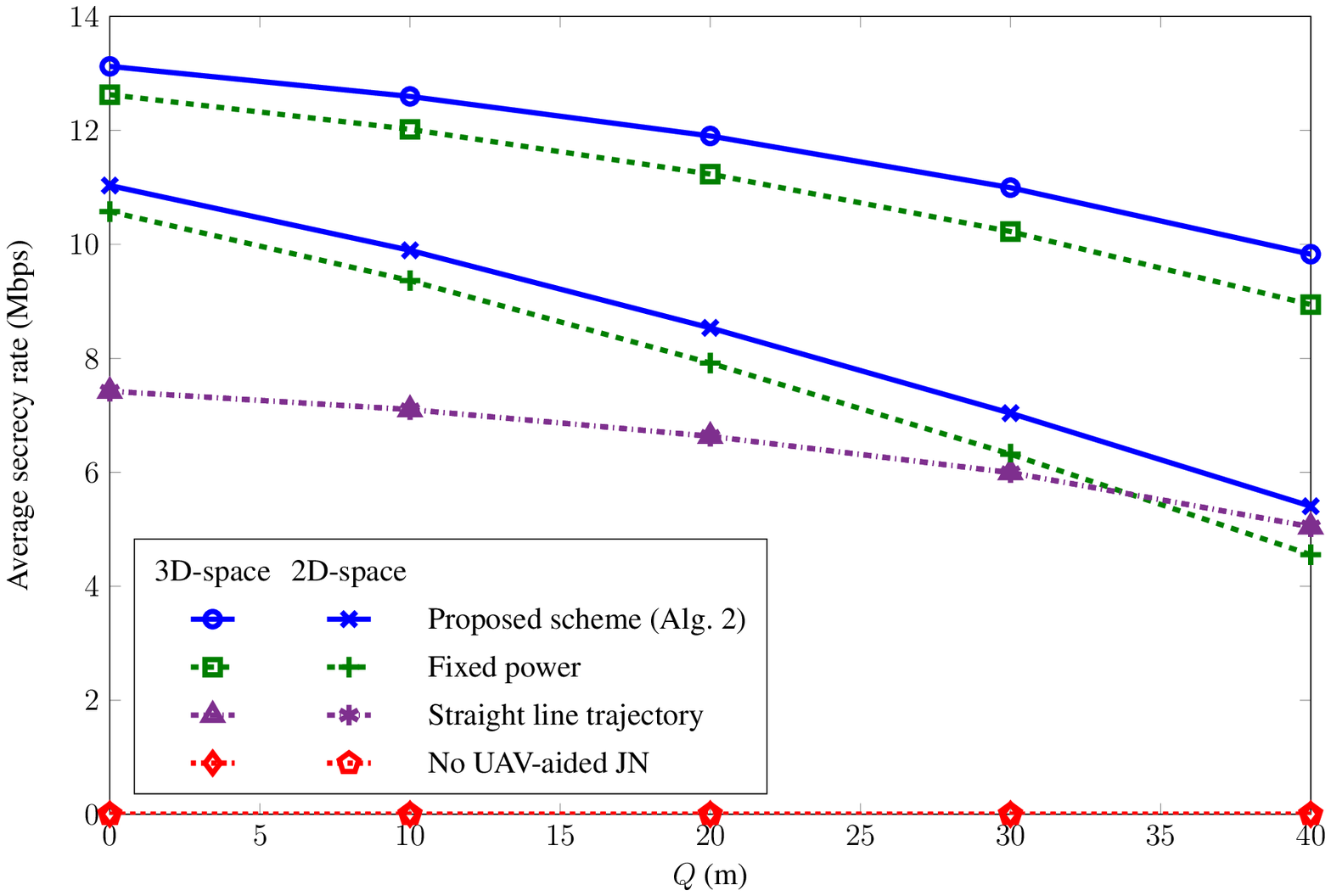}
	\end{center}
	\caption{Average secrecy rate of  different schemes versus $Q$.\label{fig:Radius}}
\end{figure}
In Fig. \eqref{fig:Radius}, we plot the secrecy rate as a function of $Q\in[0,\ 40]$ m. We note that  $Q=0$ corresponds to the case of perfect location information of Eve. It can be observed that  the average secrecy rate of all schemes drops quickly when $Q$ increases. The reasons for
these results are two-fold: 1) For a larger $Q$, Eve is able to move closer to ST to wiretap confidential messages more effectively; 2) The active region of Eve becomes wider, and thus, the location information of Eve is more difficult to estimate. In this case, the use of UAV-aided JN becomes less effective. Nevertheless, the proposed scheme still achieves the best secrecy rate by jointly optimizing the  transmit power and UAV's trajectory in the 3D space.

\subsection{Convergence Behavior of Algorithms 1 and 2}

\begin{figure}[t]
	\centering
	\begin{subfigure}[Algorithm 1.]
		{
			\includegraphics[width=0.8\columnwidth,trim={0cm 0cm 0cm 0cm}]{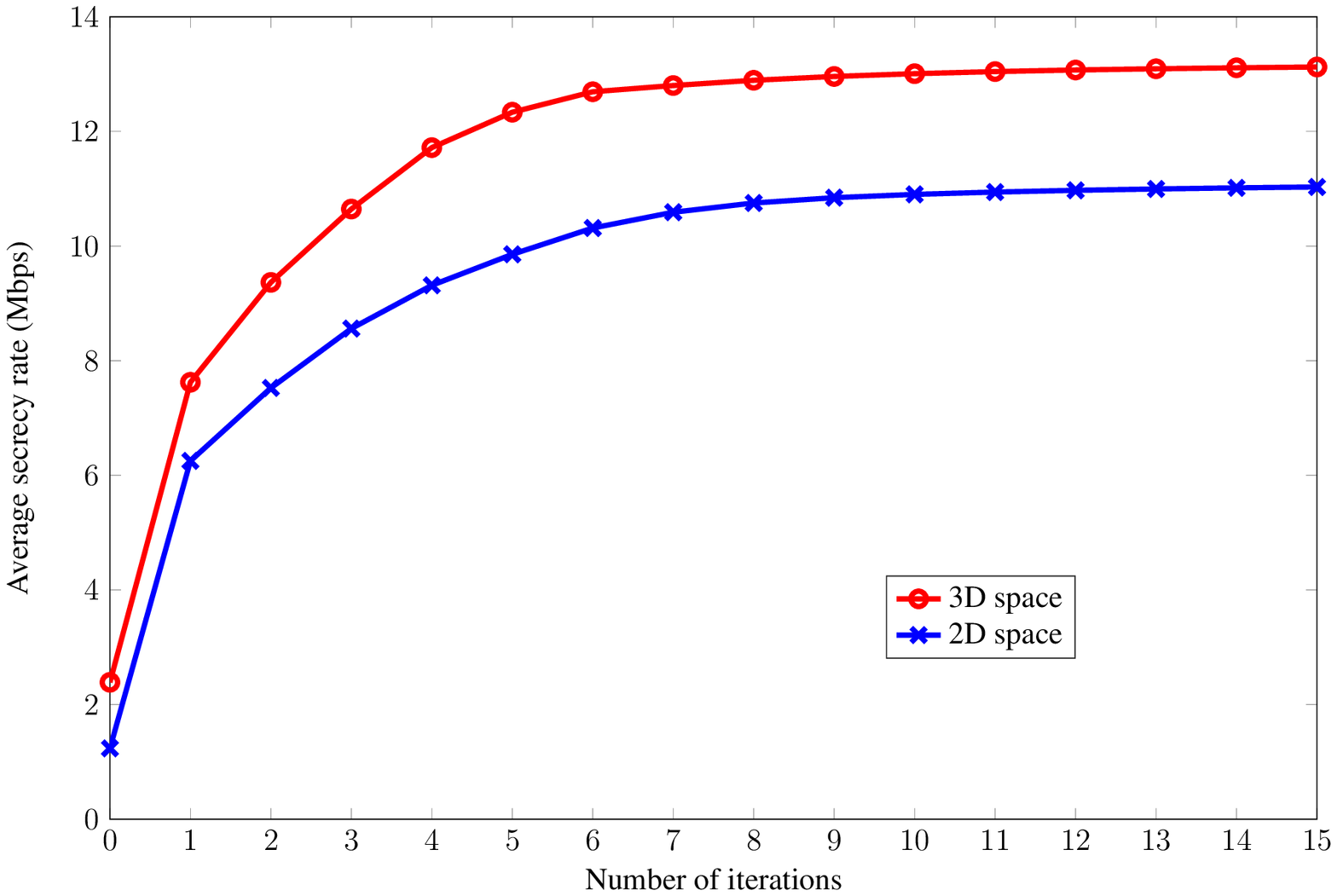}
			\label{fig:Convergence}
		}
	\end{subfigure}
	\hfill
	\begin{subfigure}[Algorithm 2.]
		{
			\includegraphics[width=0.8\columnwidth,trim={0cm 0cm 0cm 0cm}]{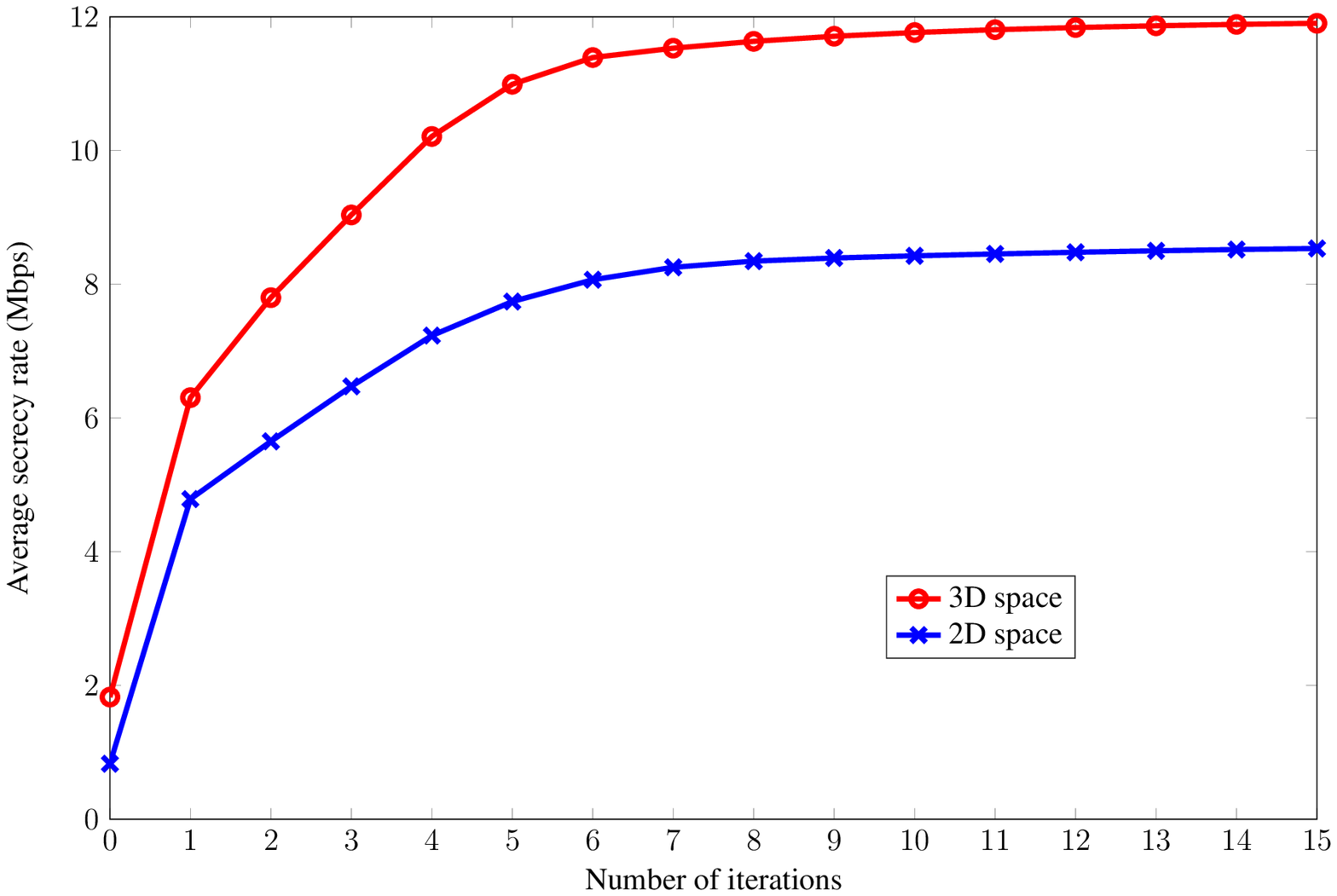} 
			
			\label{fig:Im_fig:Convergence}
		}
	\end{subfigure}
	\caption{Typical convergence behaviors of Algorithms 1 and 2 for $T = 500$s.}
	\label{fig:ConvergenceAlg12}
\end{figure}
The convergence behavior of Algorithms 1 and 2  is shown in Fig. \ref{fig:ConvergenceAlg12}, where the convergence condition is set as $\epsilon_{\mathtt{tol}}=10^{-4}$. One can see that that the proposed Algorithms monotonically improve the 
 secrecy rate after every iteration, since the optimization
variables are adjusted to find a better solution for 
next iterations. Intuitively, Algorithms 1 and 2 require only about 8 iterations to obtain the maximum secrecy rates, which are also typical for other settings.

\section{Conclusion}
This paper studied the optimization problems of maximizing the average secrecy rate of the secondary system, where a UAV is deployed to transmit JN for interfering the ST-Eve channel in both perfect and imperfect location information of Eve. The problems under the power constraints and the PR's interference power threshold are formulated as  nonconvex optimization problems. To address these problems, we first derive new nonconvex problems but with more tractable forms, and then apply  IA-based method to develop  low-complexity iterative algorithms for their solutions. Numerical results  confirmed fast convergence of the proposed algorithms  and significant performance improvement over existing schemes. They have also revealed that joint optimization of UAV's altitude (in 3D space) provides robustness against the effect of imperfect location information of Eve.


\appendices


\section{Proof of Lemma \ref{lemma 4}} \label{app D}
The worst-case secrecy rate can be written as
\begin{equation}\label{worst-case design}
\hat{R}_{\mathtt{sec}}[n] = R_{\mathtt{S}}\left[n\right] - \underset{(\Delta x,\Delta y) \in \Xi}{\maxtt} \enspace R_{\mathtt{E}}[n].
\end{equation}
Similarly to \eqref{22_app} and \eqref{27_app},  $R_{\mathtt{S}}[n]$ and $R_{\mathtt{E}}[n]$ can be safely derived as
\begin{equation}\label{worst-case design - safe form}
\hat{R}_{\mathtt{sec}}[n] = R_{\mathtt{S}}^{\mathtt{LB}}\left[n\right] - \underset{(\Delta x,\Delta y) \in \Xi}{\maxtt} \enspace R_{\mathtt{E}}^{\mathtt{UB}}[n]\bigl(\hat{d}_{\mathtt{SE}},\hat{d}_{\mathtt{UE}}[n]\bigr),
\end{equation}
where $ R_{\mathtt{S}}^{\mathtt{LB}}\left[n\right] $ is given in \eqref{22_app}. Considering $\bigl(\hat{d}_{\mathtt{SE}},\hat{d}_{\mathtt{UE}}[n]\bigr)$, $R_{\mathtt{E}}^{\mathtt{UB}}[n]$ in \eqref{27_app} can be rewritten as
\begin{align}
R_{\mathtt{E}}\left[n\right]&\leq R_{\mathtt{E}}^{\mathtt{UB}}[n]\bigl(\hat{d}_{\mathtt{SE}},\hat{d}_{\mathtt{UE}}[n]\bigr)\nonumber\\ 
&= \log_{2}\Bigl(1+\dfrac{\gamma_{0}\hat{d}_{\mathtt{SE}}^{-\varphi}p_{\mathtt{S}}\left[n\right]}{\gamma_{0}\hat{d}^{-2}_{\mathtt{UE}}[n]p_{\mathtt{U}}[n]+1}\Bigr)\label{31}.
\end{align}
Differently from the case of the perfect location information,  $\hat{d}_{\mathtt{SE}}$ is also an optimization variable of $R_{\mathtt{E}}^{\mathtt{UB}}[n]$. However, the joint optimization with $\hat{d}_{\mathtt{SE}}$ will make the optimization problem very complex. To reduce the complexity of the problem, a supremum of $R_{\mathtt{E}}^{\mathtt{UB}}[n]\bigl(\hat{d}_{\mathtt{SE}},\hat{d}_{\mathtt{UE}}[n]\bigr)$ over $\hat{d}_{\mathtt{SE}}$ is considered as
\begin{align}\label{supremum}
R_{\mathtt{E}}^{\mathtt{UB}}[n]\bigl(\ddot{d}_{\mathtt{SE}},\hat{d}_{\mathtt{UE}}[n]\bigr) = \sup_{\hat{d}_{\mathtt{SE}} \in \mathcal{D}} R_{\mathtt{E}}^{\mathtt{UB}}[n]\left(\hat{d}_{\mathtt{SE}},\hat{d}_{\mathtt{UE}}[n]\right),
\end{align}
where $\mathcal{D}$ is the set of  distances from ST to Eve; $\ddot{d}_{\mathtt{SE}}$ is the shortest distance between ST and a possible location of Eve, denoted by $ \mathbf{\ddot{c}}_{\mathtt{E}} $, such that $\mathbf{\ddot{c}}_{\mathtt{E}}\in\mathtt{C}\triangleq\{\mathbf{\hat{c}}_{\mathtt{E}}+(\Delta x_{\E}, \Delta y_{\E}, 0)|(\Delta x_{\E}, \Delta y_{\E})\in\Xi\}$. Notably,  the expression in \eqref{31} indicates that the supremum of $R_{\mathtt{E}}^{\mathtt{UB}}[n]\bigl(\hat{d}_{\mathtt{SE}},\hat{d}_{\mathtt{UE}}[n]\bigr)$ over $\hat{d}_{\mathtt{SE}}$ can be obtained by finding the minimum distance of $\hat{d}_{\mathtt{SE}}$ corresponding to $\ddot{d}_{\mathtt{SE}}$. In particular, based on geometric property, $ \mathbf{\ddot{c}}_{\mathtt{E}} $ can be easily determined as in Fig. \ref{fig:intersection}, while satisfying the condition in \eqref{supremum}. Finally, $\ddot{d}_{\mathtt{SE}}$ can be calculated as $\ddot{d}_{\mathtt{SE}}=f_{\mathtt{d}}(\mathbf{\ddot{c}}_{\mathtt{E}},\mathbf{c}_{\mathtt{ST}})$.

We should note that $\hat{d}_{\mathtt{SE}}$ and $\hat{d}_{\mathtt{UE}}[n]$ are not independent of the Eve's location $\mathbf{\hat{c}}_{\E}$. This leads to the fact that the supremum of $R_{\mathtt{E}}^{\mathtt{UB}}[n]\bigl(\hat{d}_{\mathtt{SE}},\hat{d}_{\mathtt{UE}}[n]\bigr)$ cannot be  determined only over the set of $\hat{d}_{\mathtt{SE}}$. In addition, there is  no basis to say that when Eve is located at $ \mathbf{\ddot{c}}_{\mathtt{E}} $ as shown in Fig. \ref{fig:intersection}, we can obtain the average worst-case secrecy rate. However, with a fixed point $ \mathbf{\ddot{c}}_{\mathtt{E}} $, we can compute $\ddot{d}_{\mathtt{SE}}\equiv \min_{\hat{d}_{\mathtt{SE}} \in \mathcal{D}}\{\hat{d}_{\mathtt{SE}}\} $, and then, obtain a ``strict''  worst-case secrecy rate which provides an upper bound of $R_{\mathtt{E}}^{\mathtt{UB}}[n]\bigl(\hat{d}_{\mathtt{SE}},\hat{d}_{\mathtt{UE}}[n]\bigr)$ regardless of a real location of Eve, i.e., $R_{\mathtt{E}}^{\mathtt{UB}}[n]\bigl(\ddot{d}_{\mathtt{SE}},\hat{d}_{\mathtt{UE}}[n]\bigr)\geq R_{\mathtt{E}}^{\mathtt{UB}}[n]\bigl(\hat{d}_{\mathtt{SE}},\hat{d}_{\mathtt{UE}}[n]\bigr)$. As a result, the strict worst-case objective function is derived as in \eqref{strict worst-case}.

\section{Proof of Lemma \ref{lemma 5}} \label{app E}
Constraint \eqref{estimation_error} can be rewritten  as
\begin{subnumcases}{\label{estimation_error_equi} \eqref{estimation_error}\Leftrightarrow}
\begin{bmatrix}
\Delta x_{\E} \\
\Delta y_{\E}
\end{bmatrix}^T\begin{bmatrix}
\Delta x_{\E} \\
\Delta y_{\E}
\end{bmatrix} - Q^{2} \leq 0, & \IEEEyessubnumber\label{estimation_error_equi:a}\\
f_{\mathtt{d}}(\mathbf{\hat{c}}_{\mathtt{E}}+(\Delta x_{\E}, \Delta y_{\E}, 0),\mathbf{c}_{\mathtt{U}}[n]) \leq \alpha_{\mathtt{E}}\left[n\right], \IEEEyessubnumber \label{estimation_erro_equi:b}
\end{subnumcases}
which is equivalent to the following constraints:
\begin{subnumcases}{\label{estimation_error_equi_2} \eqref{estimation_error_equi}\Leftrightarrow}
\begin{bmatrix}
\Delta x_{\E} \\
\Delta y_{\E}
\end{bmatrix}^T\begin{bmatrix}
\Delta x_{\E} \\
\Delta y_{\E}
\end{bmatrix} - Q^{2} \leq 0, & \IEEEyessubnumber\label{estimation_error_equi_2:a}\\
\begin{bmatrix}
\Delta x_{\E} \\
\Delta y_{\E}
\end{bmatrix}^T\begin{bmatrix}
\Delta x_{\E} \\
\Delta y_{\E}
\end{bmatrix} - 2\begin{bmatrix}
x_{\UAV}[n] - \hat{x}_{\mathtt{E}} \\
y_{\UAV}[n] - \hat{y}_{\mathtt{E}}
\end{bmatrix}^T\begin{bmatrix}
\Delta x_{\E} \\
\Delta y_{\E}
\end{bmatrix} \nonumber\\
 + f_{\mathtt{d}}(\mathbf{\hat{c}}_{\mathtt{E}},\mathbf{c}_{\mathtt{U}}[n]) - \alpha_{\mathtt{E}}[n] \leq 0. \IEEEyessubnumber \label{estimation_erro_equi_2:b}
\end{subnumcases}
By introducing $\theta_{\E}[n]$ such that
\begin{align}\label{slack_variable}
f_{\mathtt{d}}(\mathbf{\hat{c}}_{\mathtt{E}},\mathbf{c}_{\mathtt{U}}[n]) - \alpha_{\mathtt{E}}[n] \leq \theta_{\E}[n], \ \forall n\in\mathcal{N},
\end{align}
and applying $S$-procedure \cite{Stephen, ImperfectLocation} to \eqref{estimation_error_equi_2}, there exists
\begin{align}\label{lamda}
\mu[n] \geq 0, \ \forall n\in\mathcal{N},
\end{align}
such that 
\begin{align}\label{S-procedure}
&\begin{bmatrix}
1 & \begin{pmatrix}
\hat{x}_{\mathtt{E}} - x_{\UAV}[n] \\
\hat{y}_{\mathtt{E}} - y_{\UAV}[n]
\end{pmatrix} \\
\begin{pmatrix}
(\hat{x}_{\mathtt{E}} - x_{\UAV}[n]) & (\hat{y}_{\mathtt{E}} - y_{\UAV}[n])
\end{pmatrix} & \theta_{\E}[n]
\end{bmatrix} \nonumber\\
&\preceq 
\mu[n]
\begin{bmatrix}
1 & 0 \\
0 & -Q^{2}
\end{bmatrix}.
\end{align}
Although \eqref{S-procedure} is still intractable, we can apply Schur's complement \cite{RobustOptimization} to transform \eqref{S-procedure} into the convex constraint as
\begin{align}\label{Schur}
& &\begin{bmatrix}
1 & 0 & \hat{x}_{\mathtt{E}} - x_{\UAV}[n] \\
0 & 1 & \hat{y}_{\mathtt{E}} - y_{\UAV}[n] \\
\hat{x}_{\mathtt{E}} - x_{\UAV}[n] & \hat{y}_{\mathtt{E}} - y_{\UAV}[n] & \theta_{\E}[n]
\end{bmatrix}  \nonumber\\
&&\preceq 
\mu[n]
\begin{bmatrix}
1 & 0 & 0 \\
0 & 1 & 0 \\
0 & 0 & -Q^{2}
\end{bmatrix},
\end{align}
which is equivalent to  
\begin{align}\label{SOC}
\mathbf{S}[n]\triangleq\begin{bmatrix}
\mu[n]-1 & 0 & x_{\UAV}[n] - \hat{x}_{\mathtt{E}} \\
0 & \mu[n]-1 & y_{\UAV}[n] - \hat{y}_{\mathtt{E}} \\
x_{\UAV}[n] - \hat{x}_{\mathtt{E}} & y_{\UAV}[n] - \hat{y}_{\mathtt{E}} & -Q^{2}\mu[n]-\theta_{\E}[n]
\end{bmatrix} \succeq 
\mathbf{0}.
\end{align}
It is true that constraints \eqref{slack_variable}, \eqref{lamda} and \eqref{SOC} are convex,  and thus, the proof is completed.

\bibliographystyle{IEEEtran}
\bibliography{reference}
\end{document}